\documentclass[lettersize,onecolumn]{IEEEtran}

\usepackage{amsmath,amssymb,amsfonts}
\usepackage{amsthm}
\usepackage{bm}
\usepackage{booktabs}

\newtheorem{definition}{Definition}
\newtheorem{theorem}{Theorem}
\newtheorem{proposition}{Proposition}
\newtheorem{lemma}{Lemma}
\newtheorem{corollary}{Corollary}
\newtheorem{remark}{Remark}

\newcommand{\E}{\mathbb{E}}

\newcommand{\C}{\mathbb{C}}
\newcommand{\snr}{\rho}
\newcommand{\setX}{\mathcal{X}}

\newcommand{\setP}{\mathcal{P}}

\newcommand{\setC}{\mathcal{C}}

\DeclareMathOperator{\tr}{tr}

\title{Diversity vs.\ Degrees of Freedom\\
for Gaussian Fading Channels}

\author{Mahesh Godavarti\\Independent Researcher}

\begin{document}
\maketitle

\begin{abstract}
The classical definitions extract degrees of freedom (DOF) via $C(\snr)/\log\snr$ and diversity (DIV) via $-\log P_e(\snr)/\log\snr$, using $\log\snr$ as the common gauge for both. These ratios hide a two-step process: first, identify the gauge on which capacity or reliability actually grows; second, normalize the coefficient on that gauge by the appropriate atom. For coherent multiple-input multiple-output (MIMO) both gauges happen to be $\log\snr$ and both atom coefficients happen to be one. This paper shows that the two-step process is necessary outside this calibration case and makes it explicit using a Bhattacharyya-frontier construction. A capacity--packing sandwich theorem shows that fixed-resolution output-law packing and covering recover the capacity gauge, and a binary-endpoint theorem shows that the two-message Bhattacharyya frontier identifies the zero-rate diversity gauge. Endpoint DOF and endpoint DIV are obtained by dividing the raw coefficient on the identified gauge by the corresponding atom coefficient. For fixed deterministic channel matrix~$H$, the capacity gauge is $\log\snr$ with endpoint DOF $T\,\mathrm{rank}(H)$, while the zero-rate diversity gauge is $\snr$ with endpoint DIV $T\sigma_1^2(H)$, making fixed-$H$ a cross-gauge channel. For noncoherent scalar fast fading with $N$~receive antennas, the capacity gauge is $\log\log\snr$ with DOF~$1$, while the zero-rate diversity gauge is $\log\snr$ with endpoint DIV~$N$; the exact load-$r$ frontier gauge is $(\log\snr)^{1-r}$. The framework recovers the same-gauge cases, coherent Rayleigh MIMO and noncoherent block fading, with zero-rate DIV~$MN$. Audit tables separate exact results from lower bounds and open problems.
\end{abstract}


\section{Introduction}\label{sec:intro}

\textbf{Capacity prelog and diversity slope: intuition vs.\ measurement.}
The capacity prelog is often interpreted as the number of independent
signal dimensions visible at the output, whereas diversity is often
interpreted as the number of independent output reliability looks
available at the receiver. The classical
high-SNR definitions measure both with the fixed gauge $\log\snr$:
capacity through $C(\snr)/\log\snr$, and diversity through
$-\log P_e(\snr)/\log\snr$%
~\cite{ZhengTse2003,TarokhSeshadriCalderbank1998}.
When this gauge matches the channel, as in coherent MIMO, these
summaries are informative. When the channel lives on a different
scale, the summaries become uninformative: the intuition has not
failed, the gauge has. But even when the gauge is correct, the raw
coefficient on that gauge is not automatically a DOF or DIV count;
it becomes one only after division by the coefficient of a single
atom. The classical coherent complex-MIMO convention hides
both issues: $\log\snr$ happens to be the correct gauge, and the
relevant atom coefficients happen to be one.

\textbf{DOF and DIV as atomic counts.}
This paper uses the Godavarti--Hero intuition that DOF counts
independent channels available for communication, while DIV counts
independent reliability looks available at the
output~\cite{GodavartiHero2002}. The present formulation makes this
intrinsic by using atomic normalization. A \emph{communication atom}
is one independently controllable input degree of freedom or channel
parameter whose changes induce distinguishable output laws and can
carry distinct information. A \emph{reliability atom} is one
independent unit-normalized output reliability look. In fixed
equal-covariance AWGN, this means one unit-gain noisy output
observation carrying one unit of received pairwise signal energy.
More generally, it means one independent output-law factor
contributing the reference one-look Bhattacharyya coefficient on the
channel's zero-rate diversity gauge. A raw coefficient on the correct
gauge becomes a DOF or DIV by dividing by the corresponding atom
coefficient. In the coherent scalar and coherent MIMO calibration
cases, these atoms coincide with the usual single-input single-output
(SISO) reference atoms used in~\cite{GodavartiHero2002}. In
noncoherent block fading, the communication atom is Grassmannian
rather than a one-symbol SISO input.

\textbf{Gauge failures.}
The prelog $C(\snr)/\log\snr$ presupposes that $\log\snr$ is the
correct denominator.
Lapidoth and Moser~\cite{LapidothMoser2003} showed that noncoherent
fast fading has $C(\snr) \sim \log\log\snr$, so the classical prelog
is zero.
Lapidoth~\cite{Lapidoth2005} showed that stationary Gaussian fading
with certain non-regular spectra (e.g., the cusp family considered in
this paper) has capacity growing as $(\log\snr)^\beta$,
$0 < \beta < 1$, so the classical prelog is again zero. These
channels do not lack useful signal structure; they have zero
classical prelog because $\log\snr$ is the wrong gauge.
The same mismatch appears on the diversity side. For fixed-$H$, the
classical diversity slope is infinite because the error probability
decays exponentially in~$\snr$. The issue is again the gauge, not the
channel.

\textbf{Coefficient-normalization failures.}
Even when the gauge is correct, the raw coefficient may not equal the
DOF or DIV count. For real one-symbol AWGN, the capacity gauge is
$\log\snr$ but the raw prelog is $1/2$; the channel has one
independent communication direction, so DOF~$= (1/2)/(1/2) = 1$
after dividing by the real AWGN communication-atom coefficient $1/2$.
For fixed-$H$ AWGN, the raw diversity
coefficient is $T\sigma_1^2(H)/\ln 2$; dividing by the fixed-AWGN
reliability-atom coefficient $1/\ln 2$ gives endpoint DIV
$T\sigma_1^2(H)$. For fixed~$H$, this is an optimized output-look
count, not an unweighted count of nonzero singular directions. For
unnormalized~$H$, this endpoint DIV is gain-dependent.
In coherent complex MIMO, these normalization constants are~$1$, so
the raw coefficient and the DOF or DIV agree numerically and the
normalization step is invisible.

\textbf{This paper's approach.}
The purpose of this paper is to clarify the definitions of DOF and
DIV by separating two operations that are often merged in the
classical coherent-MIMO convention: gauge selection and atomic
normalization. Rather than assume the $\log\snr$ gauge and read its
coefficient, the framework performs two-sided gauge selection:
fixed-separation output-law packing identifies the capacity gauge on
the rate side, while the binary Bhattacharyya endpoint identifies the
zero-rate diversity gauge on the reliability side. The resulting raw
coefficients are then converted into endpoint DOF and endpoint DIV by
atomic normalization. The common tool is the Bhattacharyya frontier, defined formally in
Section~\ref{sec:tool}. For each SNR~$\snr$ and codebook
size~$K$, the frontier $\Delta_B^*(K;\snr)$ is the maximum, over all
admissible $K$-point one-block codebooks, of the minimum pairwise
Bhattacharyya distance between the induced output laws.
Fixed-separation packing gives the rate-side output-law packing
scale, which Theorem~\ref{lem:sandwich} converts to the capacity gauge
when covering and overlap conditions are met. The $K=2$ binary
endpoint gives the zero-rate diversity endpoint, and load-$r$
frontiers describe positive-load frontier behavior. The formal vocabulary in Section~\ref{subsec:vocab} defines
``diversity gauge,'' ``zero-rate diversity gauge,''
and ``endpoint DIV.''

\section{Vocabulary and Notation}\label{sec:vocab}

\subsection{Basic Notation}

All logarithms are base~2 unless otherwise noted; $\ln$
denotes the natural logarithm. We write
$\mathcal{CN}(\mu,\Sigma)$ for the circularly symmetric complex
Gaussian distribution with mean~$\mu$ and covariance~$\Sigma$.
For Hermitian matrices, $A \succeq 0$ means that $A$ is positive
semidefinite. We write $f \sim g$ for $f/g \to 1$,
$f \asymp g$ for $0 < \liminf f/g \le \limsup f/g < \infty$, and
$f = o(g)$ for $f/g \to 0$. All asymptotics are as
$\snr\to\infty$.

\subsection{Channel and Code Conventions}\label{sec:channel_model}

The examples below are Gaussian fading models built from the common
principle ``signal multiplied by a fading operator plus Gaussian
noise.'' The precise orientation of the fading operator, the block
dimension, and the meaning of one model block vary by model and are
specified in the corresponding subsection. A representative coherent
MIMO form is
\begin{equation}\label{eq:general_channel}
Y = \sqrt{\snr}\,H\,X + Z, \qquad Z \sim \mathcal{CN}(0,I),
\end{equation}
where $\snr > 0$ is the signal-to-noise ratio,
$H \in \C^{N \times M}$ is the channel matrix ($M$~transmit
antennas, $N$~receive antennas), $X \in \C^{M \times T}$ is the
transmitted signal over one $T$-slot block, and
$Z \in \C^{N \times T}$ has i.i.d.\
$\mathcal{CN}(0,1)$ entries. For noncoherent block fading we use the
conventional orientation $Y=\sqrt{\snr}\,XH+Z$, with
$X\in\C^{T\times M}$ and $H\in\C^{M\times N}$.

\textbf{Encoder/decoder.}
An $(n, K, \snr)$ code consists of a codebook
$\setC = \{x_1,\dots,x_K\}$ of $K$~codewords, each subject to the
input constraint, and a decoder mapping $Y^n$ to an estimate
$\hat{W} \in \{1,\dots,K\}$. The average error probability is
$P_e = (1/K)\sum_{i=1}^K \Pr[\hat{W} \ne i \mid W = i]$. The
rate is $R_{\mathrm{block}} = (\log K)/n$ bits per model block.

\textbf{Coding length and block normalization.}
The coding length~$n$ is the number of independent model blocks used
by the code. The Bhattacharyya frontier
(Definition~\ref{def:div_frontier}) is defined for one model block,
i.e.\ for $n = 1$. For a repetition code that transmits the same one-block codeword
across all $n$~blocks, multiplicativity (Lemma~\ref{lem:mult}) gives
pairwise distances equal to $n\,d_B^{(1)}$, so the gauge class
in~$\snr$ is unchanged. The fully optimized $n$-block frontier is a
separate object and is not used in this paper unless explicitly
defined.

For scalar fast-fading models, one model block is one physical
symbol. For fixed-$H$, coherent MIMO, and noncoherent block-fading
models, one model block is one $T$-slot input matrix. Thus $T$ is
the number of physical symbol slots inside one model block, while
$n$ is the coding length measured in independent model blocks.

The block rate is
\[
  R_{\mathrm{block}} = \frac{\log K}{n}
\]
bits per model block. For a $T$-slot model, the corresponding
per-physical-symbol rate is
\[
  R_{\mathrm{sym}} = \frac{1}{T}\,R_{\mathrm{block}}
  = \frac{\log K}{nT}.
\]
All frontier formulas are block-level unless explicitly divided
by~$T$. Whenever a conventional benchmark is quoted per physical
symbol, the text states this explicitly.

\textbf{Capacity per model block.}
For a memoryless sequence of model blocks, the one-model-block
Shannon capacity is
\[
C_{\mathrm{block}}(\snr) = \sup_{P_X} I(X;Y),
\]
where $X$ and $Y$ denote one model-block input and output. This is
ordinary Shannon capacity measured per model block, not a
finite-blocklength one-use reliability statement. Coding over~$n$
independent model blocks gives rates measured in bits per model block;
for a $T$-slot model, the corresponding per-physical-symbol rate is
obtained by dividing by~$T$.

\textbf{Atom and block conventions.}
DOF and DIV are endpoint atomic counts. A $T$-slot model block
contains $T$~physical symbol slots, so block-level DOF/DIV retain
the appropriate $T$~factor. Per-physical-symbol reporting is obtained
by dividing the block-level quantity by~$T$. The communication and reliability atoms are specified model by model. Coherent AWGN/Rayleigh
examples use the corresponding SISO reference atoms, while
noncoherent block fading uses Grassmannian communication atoms.

\begin{center}
\small
\begin{tabular}{@{}ll@{}}
\toprule
Quantity & Meaning \\
\midrule
$T$ & physical symbol slots inside one model block \\
$n$ & number of independent model blocks in the code \\
$K$ & number of messages \\
$R_{\mathrm{block}}$ & $(\log K)/n$, bits per model block \\
$R_{\mathrm{sym}}$ & $(\log K)/(nT)$, bits per physical symbol \\
$\Delta_B^*(K;\snr)$ & single-block frontier \\
$n\,\Delta_B^*(K;\snr)$ & exponent scale for $n$ independent blocks \\
\bottomrule
\end{tabular}
\end{center}

\textbf{Input constraints.}
We consider peak power ($\|X\|^2 \le P$ or $|X_t| \le \sqrt{P}$
per symbol) and average power ($\E[\|X\|^2] \le nP$). Unless stated
otherwise, the frontier calculations use the stated peak-power
constraint with $P = 1$. Peak-versus-average gauge equivalence is
proved explicitly for the scalar fast-fading packing complexity in
Lemma~\ref{lem:constraint} and is immediate up to constants for the
equal-covariance fixed-$H$ AWGN calculation. For the remaining
models, we use the constraint stated in the relevant subsection; no
general peak-versus-average invariance theorem is claimed.

For fixed-$H$ AWGN, the constraint $\|X\|_F^2\le T$ is a total block
peak-power constraint. It is not a separate unit-energy budget for
every transmit coordinate, antenna, or singular mode. Therefore
transmit directions share energy, whereas repeated receive
observations of the same transmitted signal can create additional
unit-energy output reliability looks.

\subsection{Gauge and Atom Vocabulary}\label{subsec:vocab}

We use the following terms throughout.

\textbf{Standing convention.} The default diversity-side object in
this paper is the binary endpoint of the Bhattacharyya frontier.
Therefore, unless another diversity-side benchmark is explicitly
named, \emph{raw diversity coefficient} means the coefficient of
$\Delta_B^*(2;\snr)$ on the zero-rate diversity gauge, before
reliability-atom normalization. Coefficient names do not repeat
``Bhattacharyya frontier'' unless the sentence is specifically about
the frontier as an object or a proof status.

\begin{itemize}
\item \textbf{Gauge}: a gauge is the asymptotic scale on which a
quantity grows as $\snr\to\infty$. Examples in this paper include
$\log\snr$, $\log\log\snr$, $(\log\snr)^\beta$, and~$\snr$.
A gauge is unique up to asymptotic equivalence within Hardy's
logarithmico-exponential (LE) class
(Appendix~\ref{app:gauge_unique}).

\item \textbf{Capacity gauge} $g(\snr)$: a scale on which capacity
grows, meaning $C(\snr) \asymp g(\snr)$; when a sharper coefficient
is known we write $C(\snr) \sim c\,g(\snr)$.

\item \textbf{Rate load} $r$: once the capacity gauge is $g(\snr)$,
a rate $R = r\,g(\snr)$ has coefficient~$r$ on that gauge. If
$C(\snr) \sim c\,g(\snr)$, then the corresponding fraction of
capacity is $r/c$. Thus $r$ is a capacity fraction only in
normalizations where $c = 1$. The notation matches the usual
multiplexing-gain convention.

\item \textbf{Reliability gauge}: a reliability gauge is the
high-SNR scale on which a specified error-exponent quantity grows.
The quantity must be clear from context. Examples in this paper
include the binary endpoint $\Delta_B^*(2;\snr)$, the load-$r$
frontier $\Delta_B^*(K_r(\snr);\snr)$, and operational benchmarks
such as the DMT exponent when they are explicitly named.

\item \textbf{Diversity gauge}: a diversity gauge is a reliability
gauge used for a diversity-side error exponent. The phrase is not
free-standing; it refers to the gauge of a specified diversity-side
quantity. In this paper, unless an external benchmark such as DMT is
explicitly named, ``diversity gauge'' means the zero-rate diversity
gauge defined next.

\item \textbf{Zero-rate diversity gauge}: the zero-rate diversity
gauge is the gauge $\psi_0(\snr)$ of the binary endpoint:
\[
  \Delta_B^*(2;\snr) \asymp \psi_0(\snr).
\]
If the sharper asymptotic
$\Delta_B^*(2;\snr) \sim c_{\mathrm{div}}\,\psi_0(\snr)$ holds,
then $c_{\mathrm{div}}$ is the raw diversity coefficient.
Unless an external benchmark such as DMT is explicitly named,
``zero-rate diversity gauge'' means this binary-endpoint gauge.

\item \textbf{Same-gauge / cross-gauge}: a channel is
\emph{same-gauge} if its capacity gauge and zero-rate diversity
gauge are asymptotically equivalent. It is \emph{cross-gauge} if
they are not.

\item \textbf{Raw capacity coefficient}: the coefficient of
$C(\snr)$ on the capacity gauge before atomic normalization.

\item \textbf{Raw diversity coefficient}: the coefficient of the
binary endpoint $\Delta_B^*(2;\snr)$ on the zero-rate diversity gauge
before reliability-atom normalization. If another diversity-side
benchmark is explicitly named, then ``raw diversity coefficient''
refers locally to the coefficient of that named benchmark on its
stated diversity gauge.

\item \textbf{Communication atom}: one independently controllable
input degree of freedom or channel parameter whose changes induce
distinguishable output laws and can carry distinct information. In
coherent linear channels this is an input direction. In fast fading
it is an energy or scale-family parameter. In noncoherent block
fading it is Grassmannian. The raw coefficient contributed by one
such atom on the relevant capacity gauge is the communication-atom
coefficient, denoted $\alpha_{\mathrm{comm}}$.

\item \textbf{Reliability atom}: one independent unit-normalized
output reliability look. Formally, a reliability atom is an
independent output-law factor whose pairwise Bhattacharyya distance
contributes the reference one-look coefficient on the relevant
zero-rate diversity gauge. In fixed equal-covariance AWGN, this means
one unit-gain noisy output observation carrying one unit of received
pairwise signal energy. In covariance or scale-family fading models,
the same concept is implemented by the corresponding
one-output-factor Bhattacharyya coefficient on the appropriate
diversity gauge. The raw coefficient contributed by one such atom is
the reliability-atom coefficient, denoted
$\alpha_{\mathrm{rel}}$.

\item \textbf{Atomic normalization}: the conversion from a raw
coefficient to a DOF or DIV by dividing by the corresponding atom
coefficient:
\[
\mathrm{DOF}_B = \frac{\text{raw rate-side coefficient}}
  {\alpha_{\mathrm{comm}}},
\qquad
\mathrm{DIV}_{B,0} = \frac{\text{raw diversity coefficient}}
  {\alpha_{\mathrm{rel}}}.
\]
When the rate-side coefficient is obtained directly from $C(\snr)$,
it is the raw capacity coefficient. When it is obtained from
fixed-separation packing, it is the raw rate-side endpoint
coefficient.
In coherent AWGN/Rayleigh examples, the corresponding atom is the
usual SISO reference atom. In noncoherent block fading, the
rate-side atom is a Grassmannian communication atom.

\item \textbf{Independent input dimension}: the linear-channel
realization of a communication atom. It is an
independently controllable input direction whose changes induce
distinguishable output laws and can carry distinct information. In a
fixed deterministic MIMO channel
$Y = \sqrt{\snr}\,HX + Z$, this number is
$r_H = \mathrm{rank}(H)$ per physical symbol and $Tr_H$ over one
$T$-slot block. Thus $N$~independently driven parallel input
coordinates give $N$~independent input dimensions, whereas one
scalar input observed by $N$~receive branches gives one independent
input dimension and $N$~independent unit-normalized output
reliability looks.

\item \textbf{Independent look}: the channel realization of a
reliability atom. It is not merely a physical output coordinate,
receive antenna, path, or nonzero singular direction. It is an
independent output-law factor counted after normalization to a
unit-output reliability atom. In fixed equal-covariance AWGN, this
normalization is literal: one look is one unit-energy, unit-gain
output observation. In scale-family or covariance channels, one look
means the model-specific one-output-factor Bhattacharyya contribution
on the relevant diversity gauge. A look helps reliability, while an
independent input dimension carries distinct information. In
Bhattacharyya distance, independent output-law factors add by product
multiplicativity.

\item \textbf{Output-look count}: in fixed equal-covariance AWGN, the
number of unit-energy, unit-gain independent looks generated by
a pair $X_1,X_2$. For $D=X_1-X_2$, define
\[
L_{\rm out}(D;H):=\frac14\|HD\|_F^2.
\]
Then
\[
d_B(X_1,X_2)=\frac{\snr}{\ln 2}L_{\rm out}(D;H).
\]
Thus $L_{\rm out}$ counts the independent looks contributed by
the pair on the fixed-AWGN $\snr$ gauge.
\emph{Example: $N$~receive antennas vs.\ $N$~independent channels.}
Consider $1 \times N$ fixed-AWGN with unit-gain channel vector
$H = \mathbf{1}_N$ and antipodal pair $D = X_1 - X_2$.
Each receive branch observes the same signal at full SNR and
contributes $\frac{1}{4}|D|^2$ to $L_{\rm out}$, so
$L_{\rm out} = N \cdot \frac{1}{4}|D|^2$: $N$~independent looks.
Now consider $N$~independent scalar channels sharing the same total
power~$P$.  Each channel receives $P/N$, so the antipodal pair on
channel~$k$ has $|D_k|^2 = 4P/N$ and contributes
$\frac{1}{4}\cdot 4P/N = P/N$ to the look count.  The total is
$N \times P/N = P$: exactly one independent look at full power.
Receive antennas replicate the observation without splitting the
transmitted energy; independent channels divide it.

\item \textbf{Fixed-separation packing}: fixed-separation packing
means packing codewords at a fixed positive Bhattacharyya distance
threshold $\delta > 0$, independent of~$\snr$. It is the rate-side
endpoint operation used to identify capacity gauges through packing
complexity.

\item \textbf{Packing complexity}: the log of the maximum number of
codewords that can be packed while maintaining a minimum
Bhattacharyya distance~$\delta$
(Definition~\ref{def:pack}).

\item \textbf{Bhattacharyya frontier}
$\Delta_B^*(K;\snr)$: the paper's size-$K$ separation frontier. It is
the maximum, over all admissible size-$K$ one-block codebooks at
SNR~$\snr$, of the minimum pairwise Bhattacharyya distance between the
output laws induced by distinct codewords. It is defined formally in
Definition~\ref{def:div_frontier} and is the central
gauge-identification tool in this paper. When the context is clear,
``frontier'' means the Bhattacharyya frontier.

\item \textbf{Binary endpoint}: the Bhattacharyya frontier evaluated
at $K=2$, namely $\Delta_B^*(2;\snr)$. It determines the zero-rate
diversity gauge.

\item \textbf{Endpoint DOF/DIV}: endpoint DOF and endpoint DIV are
the atom-normalized quantities obtained from the two endpoint uses of
the frontier. Fixed-separation packing gives endpoint DOF after
communication-atom normalization. The binary endpoint gives endpoint
DIV after reliability-atom normalization. If, for fixed $\delta > 0$,
\[
  K_{\mathrm{pack}}(\delta;\snr) \sim c_{\mathrm{pack}}\,g(\snr),
\]
where $g(\snr)$ is the capacity gauge, then
\[
  \mathrm{DOF}_B = \frac{c_{\mathrm{pack}}}{\alpha_{\mathrm{comm}}}.
\]
If
\[
  \Delta_B^*(2;\snr) \sim c_{\mathrm{div}}\,\psi_0(\snr),
\]
where $\psi_0(\snr)$ is the zero-rate diversity gauge, then
\[
  \mathrm{DIV}_{B,0}
  = \frac{c_{\mathrm{div}}}{\alpha_{\mathrm{rel}}}.
\]
In prose, $\mathrm{DOF}_B$ is called endpoint DOF and
$\mathrm{DIV}_{B,0}$ is called endpoint DIV.
For fixed deterministic~$H$,
\[
  \Delta_B^*(2;\snr)
  = \frac{T\sigma_1^2(H)}{\ln 2}\,\snr.
\]
Since the fixed-AWGN reliability atom has coefficient $1/\ln 2$, the
endpoint DIV is $\mathrm{DIV}_{B,0} = T\sigma_1^2(H)$.
Equivalently, the fixed-AWGN output-look count of a pair is
\[
L_{\rm out}(D;H)=\frac{1}{4}\|HD\|_F^2.
\]
The binary endpoint maximizes this output-look count under the total
block peak-power constraint. Thus $T\sigma_1^2(H)$ is not an
unweighted count of singular directions; it is the maximum number of
unit-energy, unit-gain output looks generated by an admissible binary
pair.
This endpoint DIV is gain-dependent for unnormalized~$H$.

\item \textbf{Load-$r$ frontier}: for capacity gauge
$g(\snr)$, define $K_r(\snr)=\lceil 2^{r g(\snr)}\rceil$.
The load-$r$ frontier is $\Delta_B^*(K_r(\snr);\snr)$. Its gauge is
the load-$r$ frontier gauge. This positive-load curve is separate
from endpoint DOF/DIV extraction.

\end{itemize}

The coefficients in the following table are reference coefficients of
one atom; they are not channel-specific DOF/DIV conclusions.
A raw coefficient becomes a DOF or DIV only after the relevant atom
coefficient is specified and the corresponding atomic normalization is
performed.

\begin{center}
\small
\begin{tabular}{@{}lll@{}}
\toprule
Atom & Gauge & Atom coefficient \\
\midrule
one real AWGN communication atom & $\log\snr$ & $1/2$ \\
one complex AWGN communication atom & $\log\snr$ & $1$ \\
one scalar fast-fading communication atom & $\log\log\snr$ & $1$ \\
one fixed-AWGN reliability atom (one unit-energy, unit-gain output look) & $\snr$ & $1/\ln 2$ \\
one Rayleigh fading reliability atom & $\log\snr$ & $1$ \\
one fast-fading scale-family reliability atom & $\log\snr$ & $1/2$ \\
one noncoherent Grassmannian communication atom & $\log\snr$ & $1$ \\
one noncoherent block-fading reliability atom & $\log\snr$ & $1$ \\
\bottomrule
\end{tabular}
\end{center}

The phrase ``unit-energy, unit-gain output look'' is literal for
fixed equal-covariance AWGN\@. For covariance and scale-family
channels, the reliability atom is the model-specific one-output-law
factor normalized by the corresponding one-look coefficient in the
table.

The SISO-reference interpretation applies to the coherent scalar/MIMO
reference rows. The noncoherent block-fading row uses the
Grassmannian atom convention.

\textbf{Relation to standard terminology.}
On the rate side, \emph{capacity prelog} and \emph{channel degrees of
freedom} refer to the $\log\snr$-coefficient of capacity in classical
MIMO settings. In the Zheng--Tse diversity--multiplexing tradeoff
(DMT), the \emph{multiplexing gain} refers to the coefficient~$r$ in
the rate scaling $R = r\log\snr$. In coherent
MIMO, the maximum multiplexing gain equals the channel degrees of
freedom, but the terms refer to different objects. On the diversity
side, \emph{diversity order} and \emph{diversity gain} both denote
the exponent~$d$ in $P_e \doteq \snr^{-d}$, which presupposes the
$\log\snr$ gauge. The \emph{reliability function} $E(R)$ measures
error decay in blocklength at fixed rate and SNR, whereas diversity
measures error decay in SNR.

This paper separates the two steps combined in the coherent-MIMO
convention: gauge selection and atomic normalization. In coherent
complex MIMO, both are numerically invisible because the correct gauge
is $\log\snr$ and the relevant complex AWGN/Rayleigh atom
coefficients are one. In real AWGN and fast-fading scale families, raw
coefficients contain reference constants such as $1/2$ or $1/\ln 2$,
which are removed by atom normalization.

\section{Main Contributions and Relation to Prior Work}\label{sec:contributions}

\subsection{Main Contributions}

This paper makes four contributions.

\emph{First,} it gives a two-step framework for DOF and DIV
extraction in Gaussian fading channels: identify the correct high-SNR
gauge, then normalize the raw coefficient by the appropriate
communication or reliability atom.

\emph{Second,} it introduces the Bhattacharyya frontier as the common
extraction tool for both sides of the problem.  The frontier is a
gauge-free geometric object.  Fixed-separation packing identifies the
capacity gauge and endpoint DOF after communication-atom
normalization; the binary endpoint identifies the zero-rate diversity
gauge and endpoint DIV after reliability-atom normalization.

\emph{Third,} it proves exact cross-gauge results for two central
examples.  For fixed deterministic~$H$, the capacity gauge is
$\log\snr$ with block endpoint DOF $T\,\mathrm{rank}(H)$, whereas the
zero-rate diversity gauge is~$\snr$ with endpoint DIV
$T\sigma_1^2(H)$.  For noncoherent fast fading, the capacity gauge is
$\log\log\snr$, the binary endpoint has raw coefficient $N/2$ on the
$\log\snr$ gauge, and the exact load-$r$ frontier is
$(N/2)(\log\snr)^{1-r}$.

\emph{Fourth,} it recovers the classical same-gauge calibration
cases---coherent Rayleigh MIMO and noncoherent block fading---and
provides a status audit for multipath, parallel fractional-log, and
stationary Toeplitz models, distinguishing exact results from
lower-bound, conditional, and open cases.

\subsection{Relation to Prior Work}

This paper builds on and is positioned relative to several strands.
(i)~\emph{Classical $\log\snr$ calibration.}
Telatar's coherent MIMO capacity~\cite{Telatar1999}, the Zheng--Tse
DMT~\cite{ZhengTse2003}, and the rank
criterion~\cite{TarokhSeshadriCalderbank1998} provide the coherent
complex-MIMO calibration in which both steps are numerically hidden:
$\log\snr$ is the correct gauge, and the relevant atom coefficients
are one.
(ii)~\emph{Non-$\log\snr$ capacity.}
Lapidoth--Moser~\cite{LapidothMoser2003},
Koch--Lapidoth~\cite{KochLapidoth2009}, and
Lapidoth~\cite{Lapidoth2005} provide the rate-side asymptotics on
which the non-$\log\snr$ gauges rest; this paper develops the
corresponding diversity-side analysis where exact frontier
calculations are available, and records lower-bound, conditional,
and open cases separately.
(iii)~\emph{Information dimension.}
R\'enyi's information dimension~\cite{Renyi1959} and the
Wu--Verd\'u framework~\cite{WuVerdu2010,WuShamaiVerdu2015} extract
dimensional quantities from entropy scaling; our capacity gauge is
related in spirit, but we additionally track a diversity gauge.
(iv)~\emph{Grassmannian and optical signal dimensions.}
Grassmannian signal design for block
fading~\cite{ZhengTse2002} and optical signal-dimension
intuition~\cite{PiestunMiller2000,Miller2000} are special cases
restricted to the $\log\snr$ gauge.
(v)~\emph{Godavarti--Hero DOF/DIV normalization.}
Godavarti--Hero~\cite{GodavartiHero2002} treated DOF and DIV as
SISO-normalized relative quantities and computed examples. In the
present terminology, the SISO reference is the Gaussian MIMO
realization of an atom. The present paper adds the missing
gauge layer and the atomic-normalization language: identify the
gauge, compute the raw coefficient, and
divide by the corresponding atom coefficient when DOF/DIV language
is used. New here are the frontier tool; exact
endpoint/frontier calculations for the fixed-$H$ and fast-fading
cross-gauge cases; exact zero-rate endpoint calculations for the
coherent Rayleigh and noncoherent block-fading calibration cases;
a guarded lower bound for multipath fading; a conditional parallel
fractional-log Gallager benchmark; and the explicit identification
of the stationary Toeplitz diversity gauge as open.

The frontier should be read as a
union--Bhattacharyya/packing quantity. Its advantage is gauge
stability: it is well-defined and finite on every channel studied
here, and it cleanly separates gauge identification, raw coefficient
recovery, and atomic-normalized endpoint DOF/DIV recovery. The status
labels in Table~\ref{tab:audit} record how much of this information
is recovered in each case.

\subsection{Roadmap}

The remainder of the paper motivates the framework through seven
channel families, defines the Bhattacharyya frontier, analyzes the
models channel by channel, and closes with the audit tables and open
problems.

\section{Motivation Through Examples}\label{sec:intuition}

We now motivate the gauge framework through the same channel families
analyzed later, in the same order. Each example compares the usual
$\log\snr$ gauge with the scale on which the relevant high-SNR
quantity actually grows. Where the gauge matches, the classical
summary is informative; where it does not, the summary misses the
phenomenon.

Each example below separates three layers: the mechanism producing the
raw coefficient, the gauge on which the coefficient lives, and the
atomic-normalized endpoint reading. Rate-side atoms appear as input dimensions or input parameters;
diversity-side atoms appear as unit-normalized output reliability
looks.

\subsection{Fixed-$H$: Input Dimensions vs.\ Unit Output Looks}
\label{subsec:int_det}

Consider $Y = \sqrt{\snr}\,HX + Z$, where $H \in \C^{N\times M}$
is fixed and known, and set $r_H = \mathrm{rank}(H)$.

The same linear map $X \mapsto HX$ supports two different high-SNR
operations.

On the rate side, the transmitter uses the input dimensions that
survive through~$H$. There are $r_H$ such
dimensions per physical symbol and $Tr_H$ over one $T$-slot block.
Covering the noiseless image at noise resolution gives approximately
$\snr^{Tr_H}$ distinguishable cells, hence
\[
  C_{\mathrm{block}}(\snr) \sim Tr_H\log\snr.
\]
If $H$ is full rank, this becomes $T\min(M,N)\log\snr$.

On the binary testing side, the receiver distinguishes two codewords
through unit-normalized output reliability looks at their received
separation. For $D = X_1 - X_2$, the pairwise Bhattacharyya distance
is
\[
  d_B(X_1, X_2) = \frac{\snr}{4\ln 2}\,\|HD\|_F^2
  = \frac{\snr}{4\ln 2}\sum_{t=1}^{T}\sum_{\ell=1}^{N}
    |h_\ell^\dagger d_t|^2,
\]
where $h_\ell^\dagger$ is the $\ell$-th row of~$H$ and $d_t$ is the
$t$-th column of~$D$. The summands are independent Gaussian output-law factors. After
received-energy normalization, they add as unit-energy, unit-gain
output reliability looks.

Optimizing the binary separation gives
\[
  \Delta_B^*(2;\snr) \sim \frac{T\sigma_1^2(H)}{\ln 2}\,\snr.
\]
Thus fixed deterministic~$H$ is cross-gauge: input-dimension covering
gives the $\log\snr$ capacity gauge, while unit-output-look distance expansion gives the linear~$\snr$
zero-rate diversity gauge. The
exact block-level raw coefficients are derived in
Section~\ref{subsec:ch_det}. After atom normalization,
fixed-separation packing gives endpoint DOF $Tr_H$, and the binary
endpoint gives endpoint DIV $T\sigma_1^2(H)$. This endpoint DIV is
gain-dependent for unnormalized~$H$.

\subsection{Coherent Rayleigh MIMO: The Log-Gauge Calibration Case}
\label{subsec:int_coh}

With Rayleigh fading ($H_{ij}\sim\mathcal{CN}(0,1)$, known to the
receiver), Telatar~\cite{Telatar1999} gives
$C_{\mathrm{sym}}(\snr) \sim \min(M,N)\log\snr$, and
Zheng--Tse~\cite{ZhengTse2003} gives the DMT $d^*(0) = MN$. Both
gauges are $\log\snr$, and both atom coefficients are~$1$, so the raw
coefficients are already the DOF and DIV\@. This is a same-gauge
calibration case where atomic normalization is numerically invisible.

\subsection{Noncoherent Block Fading: The Grassmannian}
\label{subsec:int_block}

When $H\in\C^{M\times N}$ is constant over $T$~symbols and unknown,
the signal space is the Grassmannian
$\mathrm{Gr}(M,T)$~\cite{ZhengTse2002}, giving
$C(\snr) \sim M(T-M)/T\cdot\log\snr$ per symbol for $T \ge 2M$.
The noncoherent DMT
result~\cite{ZhengTse2002Allerton} gives
$d^*(0) = MN$. As in coherent MIMO, both gauges are $\log\snr$ and
the relevant atom coefficients are~$1$, so this is another same-gauge
calibration case where atomic normalization is numerically invisible.

\subsection{Noncoherent Fast Fading: Radialization}
\label{subsec:int_fast}

When the fading coefficient $H_t\sim\mathcal{CN}(0,1)$ changes
independently every symbol and is unknown to the receiver, the
output $Y \sim \mathcal{CN}(0,\snr|x|^2+1)$ depends on~$x$ only
through $|x|^2$~\cite{LapidothMoser2003}. Fading radializes the
scalar channel studied here: the output law depends on the input
only through the energy $|x|^2$. (A vector-input analogue would
similarly depend on $\|x\|^2$. The formal load-$r$ statement appears in the channel analysis.) The
capacity is $C(\snr)\sim\log\log\snr$, so the capacity gauge is
$\log\log\snr$: the output law depends on the input only
through~$|x|^2$. In log-variance coordinates, the distinguishable
inputs lie on an interval of length $\ln(1+\snr) \sim \ln\snr$, so
the number of distinguishable input levels is on the order
of~$\log\snr$, and capacity grows as $\log\log\snr$.

The classical prelog $C(\snr)/\log\snr \to 0$ declares zero capacity
prelog, failing to register the $\log\log\snr$ capacity growth that
the channel genuinely has.

The $\log\snr$ gauge is mismatched to this channel on the capacity
side. On the diversity side, with $N$~receive antennas, each receive antenna provides one independent scale-family
output-law factor. After normalization by the one-look scale-family
Bhattacharyya coefficient, these factors are unit-normalized output
reliability looks. These $N$~looks multiply the one-antenna
Bhattacharyya exponent;
Section~\ref{subsec:ch_fast} gives the raw diversity coefficient~$N/2$.
Dividing by the fast-fading reliability-atom coefficient $1/2$
gives endpoint DIV~$N$.

\subsection{Multipath Fading: Guarded Fast-Fading Structure}
\label{subsec:int_multipath}

Koch--Lapidoth~\cite{KochLapidoth2009} showed
$C(\snr)\sim\log\log\snr$. The guarded construction reduces separated
data symbols to fast-fading-like scale-family observations along the
nonzero paths, giving a diversity-side lower bound on the $\log\snr$
gauge per effective guarded data symbol; the exact unguarded diversity
gauge is not identified here.

\subsection{Parallel Fractional-Log Fading: Independent-Subchannel Benchmark}
\label{subsec:int_parallel}

The parallel fractional-log model diagonalizes the channel into
independent scalar subchannels by assumption, with total capacity
gauge $G_J(\snr)=J(\snr)(\log\snr)^\beta$. Because the subchannels
are independent, the Gallager $E_0$ function factors coordinatewise,
giving a conditional lower-bound benchmark on this total gauge. This
is a controlled comparison case, not an exact solution of the
stationary Toeplitz problem.

\subsection{Stationary Toeplitz Fading: Fractional-Log Capacity and an Open Diversity Gauge}
\label{subsec:int_toeplitz}

Stationary non-regular Gaussian fading~\cite{Lapidoth2005} gives
$C(\snr) \asymp (\log\snr)^\beta$, $0<\beta<1$. The Toeplitz
covariance matrices do not commute across codewords, so the
coordinatewise factorization from the parallel benchmark is
unavailable. The diversity gauge remains open in this paper.

\subsection{Summary of Motivating Examples}

\begin{center}
\small
\begin{tabular}{@{}llll@{}}
\toprule
Channel & Capacity gauge & Classical $\log\snr$ & Takeaway \\
& (established) & summary & \\
\midrule
Fixed-$H$ & $\log\snr$ & diversity slope $= \infty$ & Wrong diversity gauge \\
Coherent Rayleigh MIMO & $\log\snr$ & Prelog and diversity work & Same-gauge calibration \\
Noncoherent block fading & $\log\snr$ & Prelog and diversity work & Same-gauge calibration \\
Noncoherent fast fading & $\log\log\snr$ & capacity prelog $= 0$ & Wrong capacity gauge \\
Multipath fading & $\log\log\snr$ & capacity prelog $= 0$ & Guarded diversity lower bound \\
Parallel fractional-log fading & $J(\snr)(\log\snr)^\beta$ total & fixed $\log\snr$ normalization mismatched & Conditional benchmark \\
Stationary Toeplitz fading & $(\log\snr)^\beta$ & capacity prelog $= 0$ & Diversity gauge open \\
\bottomrule
\end{tabular}
\end{center}

These examples show why the paper tracks capacity and diversity gauges
separately and why raw coefficients require atom normalization before
they are read as DOF or DIV. In the calibration cases, the $\log\snr$
gauge works and the relevant atom coefficients are one. In the
cross-gauge and fractional-log cases, the classical normalization
either misses the capacity gauge, the diversity gauge, or the atomic
meaning of the raw coefficient. The channel-by-channel analysis below
follows the same order and records what is proved in each case.

\section{The Bhattacharyya Frontier}\label{sec:tool}

\subsection{Bhattacharyya Distance}

\begin{definition}[Bhattacharyya coefficient and distance]
\label{def:bhatt}
For measures $P,Q$ with densities $p,q$:
\begin{align}
B(P,Q) &:= \int \sqrt{p\,q}, \label{eq:bhatt_coeff}\\
d_B(P,Q) &:= -\log B(P,Q). \label{eq:bhatt_dist}
\end{align}
\end{definition}

The key property is multiplicativity: if $P = P_1 \times P_2$ and
$Q = Q_1 \times Q_2$, then
$d_B(P,Q) = d_B(P_1,Q_1) + d_B(P_2,Q_2)$
(Lemma~\ref{lem:mult}). Independent output-law factors add: for product output laws,
$d_B$ is the sum of the component Bhattacharyya distances. After
division by the model-specific reliability-atom coefficient, this
additive sum is interpreted as the number of unit-normalized output
reliability looks. For identical independent receive observations
this sum is $N$~times the single-look distance, and for $n$~independent model blocks it is
$n$~times the single-block distance.
This multiplicativity, together with the Bhattacharyya sandwich
(Proposition~\ref{thm:bridge}), is the reason the paper uses
Bhattacharyya distance as its frontier metric.

\subsection{Packing Frontier}

Let $\setX_{\mathrm{one}} = \setX_{\mathrm{one}}(\snr)$ denote the set of valid single-use
inputs at SNR~$\snr$ (determined by the input constraint;
see Section~\ref{sec:channel_model}).

\begin{definition}[Bhattacharyya packing number]\label{def:pack}
\emph{Peak power.}
\begin{equation}
\begin{split}
N_{\mathrm{pack}}(\delta;\snr) := \max\big\{|\setC| :{}
&\setC \subseteq \setX_{\mathrm{one}},\\
&\min_{\substack{x,x'\in\setC\\x\ne x'}}
d_B(P_{Y|X=x}^{(\snr)}, P_{Y|X=x'}^{(\snr)}) \ge \delta\big\}.
\end{split}
\end{equation}
\emph{Average power.}
\begin{equation}
\begin{split}
N_{\mathrm{pack}}^{\mathrm{avg}}(\delta;\snr) := \max\big\{|\setC| :{}
&\tfrac{1}{|\setC|}\!\sum_{x\in\setC}\!\|x\|^2 \le P,\\
&\min_{\substack{x,x'\in\setC\\x\ne x'}}
d_B(P_{Y|X=x}^{(\snr)}, P_{Y|X=x'}^{(\snr)}) \ge \delta\big\}.
\end{split}
\end{equation}
The packing complexity is
$K_{\mathrm{pack}}(\delta;\snr) := \log N_{\mathrm{pack}}(\delta;\snr)$
(and similarly for the average-power variant).
\end{definition}

The following definition formalizes the frontier
construction introduced in the abstract and vocabulary section.

\begin{definition}[Bhattacharyya frontier]\label{def:div_frontier}
For integer $K\ge 2$, the Bhattacharyya frontier is
\begin{equation}\label{eq:div_frontier_full}
\Delta_B^*(K;\snr) := \max_{\substack{\setC\subseteq\setX_{\mathrm{one}}\\
|\setC|=K}} \;\min_{\substack{x,x'\in\setC\\x\ne x'}}
d_B\big(P_{Y|X=x}^{(\snr)},\, P_{Y|X=x'}^{(\snr)}\big).
\end{equation}
\end{definition}

The packing number and the frontier are inverses:
$N_{\mathrm{pack}}(\delta;\snr) \ge K \Leftrightarrow
\Delta_B^*(K;\snr) \ge \delta$
(Lemma~\ref{lem:inverse}).

\begin{theorem}[Capacity-side bridge: capacity--packing sandwich]\label{lem:sandwich}
For a one-model-block channel family $x \mapsto P_x^{(\snr)}$, let
\[
C_{\mathrm{block}}(\snr) = \sup_{P_X} I(X;Y)
\]
be the Shannon capacity in bits per model block. Let
$g(\snr) \to \infty$. Suppose the following two fixed-resolution
conditions hold.

\emph{First}, for some $c_- > 0$, there are inputs
$x_1,\dots,x_{M_\snr}$ such that, writing
$P_i = P_{x_i}^{(\snr)}$,
\[
\log M_\snr \ge c_-\,g(\snr) + o(g(\snr)),
\]
and
\[
\log\!\sup_i \sum_{j=1}^{M_\snr} B(P_i, P_j) = o(g(\snr)).
\]

\emph{Second}, for some $c_+ < \infty$, the admissible input set can
be partitioned into $L_\snr$ cells $A_1,\dots,A_{L_\snr}$ such that
\[
\log L_\snr \le c_+\,g(\snr) + o(g(\snr)),
\]
and
\[
\sup_\ell\;\inf_{Q_\ell}\;\sup_{x \in A_\ell}
D\!\left(P_x^{(\snr)} \,\big\|\, Q_\ell\right) = o(g(\snr)),
\]
with KL divergence measured in bits.

Then
\[
c_- \;\le\;
\liminf_{\snr\to\infty} \frac{C_{\mathrm{block}}(\snr)}{g(\snr)}
\;\le\;
\limsup_{\snr\to\infty} \frac{C_{\mathrm{block}}(\snr)}{g(\snr)}
\;\le\; c_+.
\]
In particular, if the same coefficient~$c$ is achieved on both sides,
namely
\[
\log M_\snr = c\,g(\snr) + o(g(\snr)),\qquad
\log L_\snr = c\,g(\snr) + o(g(\snr)),
\]
with the same overlap and local-KL conditions, then
\[
C_{\mathrm{block}}(\snr) \sim c\,g(\snr).
\]
If the sharper bounds
\[
\log M_\snr = c\,g(\snr) + O(1),\qquad
\log L_\snr = c\,g(\snr) + O(1),
\]
hold and the aggregate Bhattacharyya overlap and local KL radius are
both $O(1)$, then
\[
C_{\mathrm{block}}(\snr) = c\,g(\snr) + O(1).
\]
Thus fixed-resolution output-law packing and covering recover not only
the capacity gauge but also the raw capacity coefficient whenever
their entropy coefficients match.
\end{theorem}
\begin{proof}
See Appendix~\ref{app:sandwich_proof}.
\end{proof}

\subsection{Endpoint DOF/DIV from the Frontier}\label{subsec:endpoint_dofdiv}

The frontier has two bridge theorems.
Theorem~\ref{lem:sandwich} is the \emph{capacity-side bridge}:
fixed-resolution output-law packing and covering identify the Shannon
capacity gauge, and in coefficient-matched cases the raw capacity
coefficient. Theorem~\ref{lem:binary_endpoint} is the
\emph{diversity-side bridge}: the optimal two-message testing error is
sandwiched between constant powers of the Bhattacharyya coefficient.
Consequently, the binary endpoint $\Delta_B^*(2;\snr)$ identifies the
actual two-message zero-rate diversity gauge whenever it diverges.
Atomic normalization is then applied separately to the raw rate-side
and diversity-side coefficients.

The frontier is used in two endpoint ways. On the rate side,
fixed-separation packing gives the output-law packing scale.
Theorem~\ref{lem:sandwich} is the capacity-side bridge: when
fixed-resolution output-law packing and covering have the same entropy
gauge, this scale is the one-model-block Shannon capacity gauge; when
their entropy coefficients match, it also recovers the raw capacity
coefficient. That raw rate-side coefficient becomes endpoint DOF after
communication-atom normalization. On the diversity side, the binary
endpoint gives the raw diversity coefficient, which becomes endpoint
DIV after reliability-atom normalization.

Suppose that for fixed $\delta > 0$,
\[
  K_{\mathrm{pack}}(\delta;\snr) \sim c_{\mathrm{pack}}\,g(\snr),
\]
where $g(\snr)$ is the rate-side gauge. Let
$\alpha_{\mathrm{comm}}$ be the communication-atom coefficient on
this gauge. The endpoint DOF is
\[
  \mathrm{DOF}_B = \frac{c_{\mathrm{pack}}}{\alpha_{\mathrm{comm}}}.
\]

Suppose that
\[
  \Delta_B^*(2;\snr) \sim c_{\mathrm{div}}\,\psi_0(\snr),
\]
where $\psi_0(\snr)$ is the zero-rate diversity gauge. Let
$\alpha_{\mathrm{rel}}$ be the reliability-atom coefficient on this
gauge. The endpoint DIV is
\[
  \mathrm{DIV}_{B,0} = \frac{c_{\mathrm{div}}}{\alpha_{\mathrm{rel}}}.
\]

In coherent AWGN/Rayleigh cases, $\alpha_{\mathrm{comm}}$ and
$\alpha_{\mathrm{rel}}$ coincide with the usual SISO reference-atom
coefficients. In noncoherent block fading, the communication-side
normalization is Grassmannian.

For $r > 0$, the frontier
$\Delta_B^*(K_r(\snr);\snr)$, with
$K_r(\snr) = \lceil 2^{r g(\snr)}\rceil$, gives the load-$r$
frontier. This positive-load curve records how the frontier gauge
changes as rate load increases. Its coefficient on the load-$r$ frontier gauge may be called the
load-$r$ coefficient. Endpoint DOF is extracted from fixed-separation
packing. Endpoint DIV is extracted from the binary endpoint.

\subsection{Frontier Roles and Status}\label{subsec:status_labels}

The following terms use Definition~\ref{def:div_frontier}.

\begin{itemize}
\item \textbf{Load-$r$ frontier gauge.} For capacity gauge $g(\snr)$,
define
\[
    K_r(\snr)=\left\lceil 2^{r g(\snr)}\right\rceil.
\]
The gauge of $\Delta_B^*(K_r(\snr);\snr)$ is the load-$r$ frontier
gauge. When matching upper and lower bounds for this quantity are
proved, the load-$r$ frontier gauge is identified. Equality with an
operational reliability gauge $\psi_{\mathrm{op},r}(\snr)$ is a
separate, model-dependent question.

\item \textbf{Zero-rate diversity gauge.} As defined in
Section~\ref{subsec:vocab}, the zero-rate diversity gauge is the gauge
$\psi_0(\snr)$ satisfying
\[
  \Delta_B^*(2;\snr) \asymp \psi_0(\snr).
\]
By Theorem~\ref{lem:binary_endpoint}, this endpoint is equivalent up to
constant factors to the best equal-prior two-message testing exponent
whenever $\Delta_B^*(2;\snr) \to \infty$. Unless an external
operational benchmark, such as a DMT value $d^*(0)$, is explicitly
cited, ``zero-rate diversity gauge'' means this binary-endpoint gauge.
We do not claim equality with the full zero-rate reliability function
or with every possible $r \downarrow 0$ operational limit.

\end{itemize}

\textbf{Gauge identification and endpoint normalization.}
The frontier has two endpoint roles and one
positive-load extension. On the capacity side, fixed-separation
packing gives a geometric output-law packing scale.
Theorem~\ref{lem:sandwich} converts this scale into the one-model-block
Shannon capacity gauge, and in coefficient-matched cases into the raw
capacity coefficient. Dividing that raw coefficient by
$\alpha_{\mathrm{comm}}$ gives endpoint DOF. On the diversity side,
the binary endpoint $\Delta_B^*(2;\snr)$ is tied to the optimal
two-message error exponent by Theorem~\ref{lem:binary_endpoint}; when
this endpoint has raw coefficient $c_{\mathrm{bin}}$, dividing by
$\alpha_{\mathrm{rel}}$ gives endpoint DIV. At positive load,
$\Delta_B^*(K_r(\snr);\snr)$ gives the load-$r$ frontier.

In this paper, the load-$r$ frontier gauge is identified exactly
only for the fast-fading model. For the other channel families,
load-$r$ frontier-gauge identification is left for future work.

\textbf{Bridge to error probability.}

\begin{theorem}[Diversity-side bridge: binary endpoint identifies the
zero-rate diversity gauge]
\label{lem:binary_endpoint}
For two channel-induced output distributions $P,Q$, let
$P_{e,2}^*(P,Q)$ be the optimal equal-prior binary testing error.
Then
\begin{equation}\label{eq:binary_sandwich}
d_B(P,Q) + 1 \;\le\; -\log_2 P_{e,2}^*(P,Q) \;\le\;
2\,d_B(P,Q) + 2.
\end{equation}
Consequently, if\/ $E_2^*(\snr) := \sup_{x\ne x'}
[-\log_2 P_{e,2}^*(P_{Y|x}^{(\snr)},P_{Y|x'}^{(\snr)})]$,
then
\begin{equation}\label{eq:frontier_operational}
\Delta_B^*(2;\snr) + 1 \;\le\; E_2^*(\snr) \;\le\;
2\,\Delta_B^*(2;\snr) + 2.
\end{equation}
Therefore, whenever $\Delta_B^*(2;\snr)\to\infty$, the binary
endpoint identifies the actual two-message zero-rate diversity gauge:
\begin{equation}\label{eq:gauge_equiv}
E_2^*(\snr) \;\asymp\; \Delta_B^*(2;\snr).
\end{equation}
The theorem identifies the gauge. Exact raw operational coefficients
require model-specific tightening or direct computation.

For the $n$-fold repetition of any fixed single-use pair, the same
inequalities hold with $d_B$ replaced by $n\,d_B$. In particular,
repeating a pair that achieves $\Delta_B^*(2;\snr)$ gives a
two-message exponent of gauge $n\,\Delta_B^*(2;\snr)$. This
statement does not upper-bound arbitrary $n$-block two-message codes
unless the corresponding $n$-block frontier is separately analyzed.
\end{theorem}
\begin{proof}
The optimal binary testing error satisfies
$P_{e,2}^* = (1-\mathrm{TV}(P,Q))/2$.
The total-variation/Bhattacharyya inequalities give
$1-B(P,Q) \le \mathrm{TV}(P,Q) \le \sqrt{1-B(P,Q)^2}$,
hence $B(P,Q)^2/4 \le P_{e,2}^* \le B(P,Q)/2$.
Taking $-\log_2$ yields~\eqref{eq:binary_sandwich}.
Optimizing over pairs $x\ne x'$
gives~\eqref{eq:frontier_operational}.
\end{proof}

\begin{proposition}[Bhattacharyya sandwich for repeated one-block
codebooks]\label{thm:bridge}
Let $\setC = \{x_1,\dots,x_K\}\subseteq\setX_{\mathrm{one}}$ be a
one-block codebook with minimum pairwise Bhattacharyya distance
$\Delta_{\min}(\setC;\snr)$. For $n$~independent repetitions of the
same one-block codebook, the optimal $K$-message error probability
satisfies
\begin{equation}\label{eq:bridge}
\frac{1}{2K}\,2^{-2n\,\Delta_{\min}(\setC;\snr)}
\;\le\; P_{e,\mathrm{opt}}(\setC^{(n)})
\;\le\; (K-1)\,2^{-n\,\Delta_{\min}(\setC;\snr)}.
\end{equation}
The upper bound is achieved by ML via the union--Bhattacharyya bound;
the lower bound follows from the closest pair and
Theorem~\ref{lem:binary_endpoint}.
\end{proposition}
\begin{proof}
See Appendix~\ref{app:coding_bounds}.
\end{proof}

\begin{corollary}\label{cor:optimal_packing}
For the frontier-optimal one-block codebook repeated over
$n$~independent blocks,
\begin{equation}\label{eq:closest_pair}
\frac{1}{2K}\,2^{-2n\,\Delta_B^*(K;\snr)}
\;\le\; P_{e,\mathrm{opt}}
\;\le\; (K-1)\,2^{-n\,\Delta_B^*(K;\snr)}.
\end{equation}
\end{corollary}
\begin{proof}
Choose $\setC$ to maximize $\Delta_{\min}(\setC;\snr)$; by
definition, $\Delta_{\min}(\setC;\snr) = \Delta_B^*(K;\snr)$.
\end{proof}

\textbf{Status labels.}

For each channel, we compare the frontier information
against the operational benchmarks. Three separate questions are
tracked:

\begin{enumerate}
\item \textbf{Gauge match}: does the frontier gauge agree with the
best operational benchmark?
\item \textbf{Raw coefficient match}: does the computed raw
coefficient agree with the named comparison benchmark, when one is
being used?
\item \textbf{Endpoint match}: after division by the appropriate atom
coefficient, does the endpoint recover the corresponding DOF or DIV?
\end{enumerate}

The following status labels, from strongest to weakest, are used:

\begin{itemize}
\item \textbf{Exact frontier}: matching upper and lower bounds
for $\Delta_B^*$, including the raw coefficient, at the stated
$K(\snr)$.
\item \textbf{Zero-rate exact}: matching upper and lower
bounds for the binary endpoint $\Delta_B^*(2;\snr)$, including the
raw diversity coefficient, are proved.
\item \textbf{Zero-rate coefficient exact}: the binary endpoint
recovers both the zero-rate gauge and the known zero-rate raw
operational coefficient.
\item \textbf{Endpoint exact}: the frontier endpoint recovers the
gauge and, after atom normalization, the corresponding DOF or DIV.
\item \textbf{Gauge only}: the frontier identifies the relevant
gauge, while raw coefficient recovery remains separate.
\item \textbf{Lower-bound only}: the paper proves a lower bound at
a stated scale.
\item \textbf{Conditional}: the conclusion depends on an explicit
hypothesis, such as the existence and positivity of a normalized
Gallager limit.
\item \textbf{Open}: the paper does not determine the gauge.
\end{itemize}

For load $r>0$, ``identified'' means matching upper and lower bounds
for $\Delta_B^*(\lceil 2^{r g(\snr)}\rceil;\snr)$ are proved.

The distinction between raw coefficient match and atomic-normalized
endpoint match matters. By Theorem~\ref{lem:binary_endpoint}
(diversity-side bridge), the binary endpoint
$\Delta_B^*(2;\snr)$ is not merely a proxy for zero-rate
reliability. It is gauge-equivalent to the optimal equal-prior
two-message testing exponent, and therefore identifies the actual
two-message zero-rate diversity gauge whenever it diverges. What may
remain model-dependent is the exact raw operational coefficient, not
the gauge.

\textbf{Scope disclaimer.}
In this paper, the frontier is evaluated model by model. For each
channel where a frontier claim is made, the gauge
equivalence is verified by direct computation. The bridge to error
probability uses the Bhattacharyya
sandwich~(Proposition~\ref{thm:bridge}), which controls
repetition-code reliability from both sides at the gauge level.
These bounds identify binary gauges and support the endpoint DIV
calculations stated in the exact cases. Raw operational reliability
coefficients and full positive-rate reliability curves require
model-specific tools, as recorded in the status labels and audit
table.

\section{Channel-by-Channel Analysis}\label{sec:channels}

Each channel follows the same template: model, intuition, operational
benchmark, frontier result, gauge recovery, and status.

\subsection{Fixed-$H$}\label{subsec:ch_det}

\textbf{Model.}
$H \in \C^{N\times M}$ is a known, deterministic matrix, and
$r_H = \mathrm{rank}(H)$. Assume $H \ne 0$. If $H = 0$, then both capacity and the binary endpoint are zero for
every~$\snr$, so the gauge statements in this subsection no longer
apply.
The input $X \in \C^{M\times T}$
satisfies $\|X\|_F^2 \le T$. Over one $T$-slot block:
\begin{equation}\label{eq:model_det}
Y = \sqrt{\snr}\,H\,X + Z, \quad Z_{ij} \sim \mathcal{CN}(0,1)
\text{ independently.}
\end{equation}

\textbf{Intuition: covering vs.\ expansion.}
The channel output (ignoring noise) is the linear map $X \mapsto HX$.
The capacity gauge and the zero-rate diversity gauge arise from two
different geometric operations on this map
(Table~\ref{tab:dof_div_dual}).

\emph{Capacity side: input dimensions.}
The image of $X \mapsto HX$ has $r_H$ independent input dimensions
per physical symbol and $Tr_H$ over one $T$-slot block; these are the
input dimensions available for carrying distinct information.
Equivalently,
$\mathrm{rank}(I_T \otimes H) = Tr_H$.
Covering the received image at AWGN noise resolution gives
\begin{equation}\label{eq:det_dof}
N_{\mathrm{cover}}(\snr) \asymp \snr^{Tr_H}.
\end{equation}
Taking logarithms gives
\begin{equation}\label{eq:dof_capacity}
C_{\mathrm{block}}(\snr) \sim Tr_H\log\snr, \qquad
C_{\mathrm{sym}}(\snr) \sim r_H\log\snr.
\end{equation}
This is an instance of Theorem~\ref{lem:sandwich}. In noise-normalized
output coordinates, set
\[
z = \sqrt{\snr}\,HX,\qquad
S_\snr = \{\sqrt{\snr}\,HX : \|X\|_F^2 \le T\}.
\]
The set $S_\snr$ is a fixed ellipsoid scaled by~$\sqrt{\snr}$ in
$d = Tr_H$ complex dimensions. Hence fixed-radius Euclidean packing
and covering numbers in $S_\snr$ both have logarithmic size
\[
d\log\snr + O(1) = Tr_H\log\snr + O(1).
\]
For equal-covariance complex Gaussians centered at $z$ and~$z'$,
\[
B(P_z, P_{z'}) = \exp\!\Big(-\frac{\|z - z'\|^2}{4}\Big),
\qquad
D(P_z \| P_{z'}) = \frac{\|z - z'\|^2}{\ln 2}.
\]
A fixed lattice packing therefore has bounded aggregate Bhattacharyya
overlap, because the corresponding Gaussian theta sum is finite.
A fixed-radius covering has bounded local KL radius.
Theorem~\ref{lem:sandwich} therefore gives the coefficient-level
capacity statement
\[
C_{\mathrm{block}}(\snr) = Tr_H\log\snr + O(1).
\]
Thus fixed-separation output-law packing recovers both the capacity
gauge $\log\snr$ and the raw block capacity coefficient~$Tr_H$.
Since the complex AWGN communication atom has coefficient~$1$ on the
$\log\snr$ gauge, the block endpoint DOF is~$Tr_H$.
Dividing by~$T$ gives the
per-physical-symbol value $r_H$.
For full-rank~$H$, $r_H = \min(M,N)$.

\emph{Diversity side: unit-normalized output reliability looks.}
For two codewords $X_1, X_2$, let $D = X_1 - X_2$. By the
same-covariance Gaussian Bhattacharyya formula,
\begin{equation}\label{eq:det_div}
d_B(X_1, X_2) = \frac{\snr}{4\ln 2}\,\|HD\|_F^2.
\end{equation}
For fixed equal-covariance AWGN, define the output-look count of the
pair by
\[
L_{\rm out}(D;H):=\frac14\|HD\|_F^2.
\]
Then
\[
d_B(X_1,X_2)=\frac{\snr}{\ln 2}L_{\rm out}(D;H).
\]
Thus $L_{\rm out}$ is the number of unit-energy, unit-gain output
reliability looks generated by the pair.

Writing $d_t$ for the $t$-th column of~$D$ and $h_\ell^\dagger$ for
the $\ell$-th row of~$H$,
\begin{equation}\label{eq:det_div_looks}
d_B(X_1, X_2) = \frac{\snr}{4\ln 2}
  \sum_{t=1}^{T}\sum_{\ell=1}^{N} |h_\ell^\dagger d_t|^2.
\end{equation}
These terms are independent Gaussian output-law factors for the same
pairwise received separation. After received-energy normalization,
they contribute additively to the number of unit-energy, unit-gain
output reliability looks.

The best binary separation uses the strongest right singular direction
of~$H$. Since $\|X_i\|_F^2 \le T$,
$\max \|H(X_1-X_2)\|_F^2 = 4T\sigma_1^2(H)$, and the binary
endpoint is
\begin{equation}\label{eq:div_exponent}
\Delta_B^*(2;\snr) = \frac{T\sigma_1^2(H)}{\ln 2}\,\snr.
\end{equation}
Equivalently,
\[
\mathrm{DIV}_{B,0}
=
\max_{\|X_i\|_F^2\le T}
\frac14\|H(X_1-X_2)\|_F^2
=
T\sigma_1^2(H).
\]
Thus $T\sigma_1^2(H)$ is the optimized unit-output-look count under
one total block power budget. It is not an unweighted count of
nonzero singular directions.

Thus $T\sigma_1^2(H)/\ln 2$ is the raw diversity coefficient.
The fixed-AWGN reliability atom has coefficient
$1/\ln 2$ on the $\snr$ gauge, so the endpoint DIV is
$T\sigma_1^2(H)$. For an unnormalized deterministic matrix~$H$, this
endpoint DIV is gain-dependent.

\emph{Examples.}
If $H = I_N$, then $r_H = N$. This is the diagonal $N$-channel
AWGN example: there are $N$ independent input dimensions, so
$C_{\mathrm{sym}}(\snr) \sim N\log\snr$,
$C_{\mathrm{block}}(\snr) \sim TN\log\snr$.
Under the total block peak-power constraint, these $N$ transmit
directions share the binary-pair energy. Therefore the maximum
unit-energy, unit-gain output-look count is $T$, not $TN$.
If $H = \mathbf{1}_N \in \C^{N\times 1}$, then
$r_H = 1$, $\sigma_1^2(H) = N$. This is the clean SIMO example:
there is one input dimension and $N$ unit-normalized output
reliability looks, because the same input energy is copied into $N$
independent output branches, so
$\Delta_B^*(2;\snr) \sim \frac{TN}{\ln 2}\,\snr$.
If $H = \mathbf{1}_N\mathbf{1}_M^\top$, then
$r_H = 1$, $\sigma_1^2(H) = MN$. Thus the rate-side calculation gives one input dimension
per physical symbol, while the binary-endpoint
calculation gives raw diversity coefficient $TMN/\ln 2$; after reliability-atom
normalization, the corresponding endpoint DIV is $TMN$.
This endpoint DIV is gain-dependent for unnormalized~$H$.
The extra factor~$M$ is coherent transmit gain: it increases received
output energy and hence the unit-output-look count; it is not an
additional count of independent receive branches.

These examples separate input dimensions from unit-normalized output
reliability looks: $H = I_N$ gives $N$~input dimensions sharing one
total binary energy budget, whereas $H = \mathbf{1}_N$ gives one
input dimension whose energy is copied into $N$~independent output
reliability looks.

\emph{Why the gauges differ.}
Covering produces logs (the log of a count); expansion produces
linear growth (a squared distance). Both operations act on the
same linear map $X \mapsto HX$, so the gauge mismatch follows
directly from the geometry.

\textbf{Operational benchmark.}
Capacity: $C_{\mathrm{block}}(\snr) \sim Tr_H\log\snr$
(metric entropy of
ellipsoids~\cite{KolmogorovTikhomirov1959}). ML error exponent:
for codeword difference $D = X_1 - X_2$,
$-\log P_e \propto \snr\,\|HD\|_F^2$ (exact Bhattacharyya
for the AWGN channel; see Appendix~\ref{app:bhatt_gaussian}).

\textbf{Frontier result.}
By Lemma~\ref{lem:bhatt_same_cov}, with
$\Sigma = I_N$ (AWGN noise covariance) and
$\mu_i = \sqrt{\snr}\,H X_i$:
\begin{equation}\label{eq:bhatt_det}
d_B(P_{Y|X_1},P_{Y|X_2})
= \frac{\snr}{4\ln 2}\|H(X_1-X_2)\|_F^2.
\end{equation}
Algebraically, this decomposes into $N\!\cdot\!\mathrm{rank}(D)$
terms, each linear in~$\snr$ with coefficient
$s_j^2\,|H_\ell^\dagger u_j|^2$, possibly zero
(Lemma~\ref{lem:rank_product}(a)); the decomposition is used in
the coherent Rayleigh analysis (Section~\ref{subsec:ch_coh}).

Let $\sigma_1 \ge \cdots \ge \sigma_{\min(M,N)}$
be the singular values of~$H$. Then
$\|HD\|_F^2 = \tr(D^\dagger H^\dagger H D)$.
Under the power constraint, the achievable difference satisfies
$\|D\|_F \le 2\sqrt{T}$ (antipodal signaling: $X_2 = -X_1$,
$\|X_i\|_F^2 \le T$). The maximum of $\|HD\|_F^2$ over this
constraint is $4T\sigma_1^2(H)$, achieved by concentrating all
energy in the top right singular direction of~$H$. Hence
$\Delta_B^*(2;\snr) = \frac{\snr}{4\ln 2}\cdot
4T\sigma_1^2(H) = \frac{T\sigma_1^2(H)}{\ln 2}\,\snr
\sim c_H\cdot\snr$,
where $c_H = T\sigma_1^2(H)/\ln 2$. This gives the zero-rate
diversity gauge~$\snr$; the raw diversity coefficient depends on $\sigma_{\max}(H)$, not on
$\|H\|_F^2$, because the binary endpoint maximizes received output
energy under one total block power budget rather than assigning a
separate unit-energy budget to every singular mode.

For the packing number, the $\delta$-separated codebook must satisfy
$\|H(X_i - X_j)\|_F^2 \ge 4\delta\ln 2/\snr$ for all $i\ne j$.
In the unscaled image space $\{HX : \|X\|_F^2 \le T\}$, the
semi-axes are proportional to the singular values $\sigma_k(H)$.
Equivalently, in the noise-normalized received image
$\{\sqrt{\snr}\,HX : \|X\|_F^2 \le T\}$, the semi-axes are
proportional to $\sqrt{\snr}\,\sigma_k(H)$. Packing this received
image at constant noise resolution is the same as packing the
unscaled image at radius $O(\snr^{-1/2})$, giving
$\log K_{\mathrm{pack}}(\delta;\snr) \sim
Tr_H\cdot\log\snr$~\cite{KolmogorovTikhomirov1959}.

\begin{theorem}[Fixed-$H$ cross-gauge endpoint]\label{thm:fixed_H}
For $H \ne 0$ with $r_H = \mathrm{rank}(H)$, the capacity gauge is
$\log\snr$ and the zero-rate diversity gauge is~$\snr$:
\begin{equation}\label{eq:fixedH_capacity}
C_{\mathrm{block}}(\snr) = Tr_H\log\snr + O(1),
\end{equation}
\begin{equation}\label{eq:fixedH_binary}
\Delta_B^*(2;\snr) = \frac{T\sigma_1^2(H)}{\ln 2}\,\snr.
\end{equation}
Block endpoint DOF is~$Tr_H$; endpoint DIV is~$T\sigma_1^2(H)$.
\end{theorem}
\begin{proof}
The capacity statement follows from Theorem~\ref{lem:sandwich}
applied to the noise-normalized ellipsoid packing
and covering, as derived above.
The binary endpoint follows from
the same-covariance formula~\eqref{eq:bhatt_det} and the
constraint-optimal antipodal pair.
\end{proof}

\emph{Cross-gauge observation and load-$r$ gauge prediction.}
At per-physical-symbol rate $R_{\mathrm{sym}} = r\log\snr$,
$0 \le r < r_H$; here $r$ is a rate-load coefficient, and the
corresponding capacity fraction is $r/r_H$.
The corresponding block rate is
$R_{\mathrm{block}} = Tr\log\snr$. Hence one $T$-slot block contains
$K = 2^{R_{\mathrm{block}}} = \snr^{Tr}$ codewords. A Euclidean
volume-packing comparison in the received image of block dimension
$Tr_H$ gives the gauge-level prediction
\begin{equation}\label{eq:det_frontier}
\Delta_B^*(K;\snr) \asymp c(H)\,\snr^{1-r/r_H}.
\end{equation}
Here $c(H)$ may depend on the singular values of~$H$ and on
codebook geometry. This is a load-$r$ frontier prediction. Matching upper and
lower bounds for the load-$r$ frontier, and comparison with an
operational reliability gauge, are not proved here and are listed as
future work.

\textbf{Gauge recovery.}
Capacity gauge: $\log\snr$ (exact). Zero-rate diversity gauge:
$\snr$ (exact). The Bhattacharyya distance equals the ML pairwise
error exponent for the AWGN channel at high SNR, so the frontier is
tight.

\textbf{Status:} capacity gauge exact; zero-rate exact; raw
diversity coefficient exact for fixed-$H$ AWGN;
endpoint DIV $T\sigma_1^2(H)$ recovered.
Classification: \textbf{cross-gauge}.

\textbf{What is known / what is new.}
Metric entropy of ellipsoids is
classical~\cite{KolmogorovTikhomirov1959}; the rank criterion is due
to~\cite{TarokhSeshadriCalderbank1998}; the optical
signal-dimension interpretation is
in~\cite{PiestunMiller2000,Miller2000}. New here:
the covering/expansion argument that produces both gauges from the
same map $X \mapsto HX$, with raw diversity coefficient
$T\sigma_1^2(H)/\ln 2$ and endpoint DIV
$T\sigma_1^2(H)$. The rank-product decomposition
(Lemma~\ref{lem:rank_product}) and its role in the genuine
$N\!\cdot\!\mathrm{rank}(D)$ diversity coefficient appear in the
coherent Rayleigh analysis (Section~\ref{subsec:ch_coh}).

\begin{table}[t]
\centering
\caption{Fixed-$H$ channel: input dimensions vs.\ unit output looks.}
\label{tab:dof_div_dual}
\small
\begin{tabular}{@{}p{1.8cm}p{3.0cm}p{3.0cm}@{}}
\toprule
& Capacity side & Diversity side \\
\midrule
Operation & Use input dimensions through~$H$ &
Use unit-normalized output looks of the pairwise received separation \\
Object & Image $\{HX : \|X\|_F^2\!\le\!T\}$ &
$L_{\rm out}(X_1\!-\!X_2;H)=\frac14\|H(X_1\!-\!X_2)\|_F^2$ \\
Additive structure & $Tr_H$ input dimensions per block &
$\sum_{t,\ell}|h_\ell^\dagger d_t|^2$ \\
Raw coefficient & $Tr_H$ block capacity coefficient &
$T\sigma_1^2(H)/\!\ln 2$ raw diversity coefficient \\
Atom coefficient & $1$ complex AWGN communication atom &
$1/\ln 2$ fixed-AWGN reliability atom \\
Endpoint DOF/DIV & block DOF $Tr_H$ &
endpoint DIV $T\sigma_1^2(H)$ \\
Gauge & $\log\snr$ & $\snr$ \\
\bottomrule
\end{tabular}

\smallskip
The capacity coefficient arises from input-dimension covering; the
raw diversity coefficient arises from optimized pairwise output
separation. The diversity-side look is normalized; it is not a bare
physical coordinate, receive antenna, or singular direction.
\end{table}

\subsection{Coherent Rayleigh MIMO}\label{subsec:ch_coh}

\textbf{Model.}
$H \in \C^{N\times M}$ with i.i.d.\
$H_{ij} \sim \mathcal{CN}(0,1)$ (Rayleigh fading), constant over
$T \ge M$ symbol slots, known to the receiver. Over one $T$-slot block:
$Y = \sqrt{\snr}\,HX + Z$ with $X \in \C^{M\times T}$,
$\|X\|_F^2 \le T$.

\textbf{Intuition.}
By the rank-product decomposition (Lemma~\ref{lem:rank_product}), for known~$H$
each output-law factor contributes $\sim\snr$ to the Bhattacharyya
distance. When $H$ is random, the deep-fade event $|H|^2\approx 0$
dominates the average, pulling each factor down to~$\sim\log\snr$
(Lemma~\ref{lem:rank_product}(b)), matching the capacity gauge.
Both capacity and diversity should live on the $\log\snr$ scale.

\textbf{Operational benchmark.}
Capacity: per physical symbol,
$C_{\mathrm{sym}}(\snr) \sim \min(M,N)\log\snr$~\cite{Telatar1999};
over one $T$-slot block,
$C_{\mathrm{block}}(\snr) \sim T\min(M,N)\log\snr$.
DMT: $d^*(r)$ is the piecewise-linear curve connecting
$(k,\,(M-k)(N-k))$, $k=0,\dots,\min(M,N)$, at per-physical-symbol rate $R_{\mathrm{sym}} = r\log\snr$,
with $P_e \doteq \snr^{-d^*(r)}$~\cite{ZhengTse2003}. Capacity
gauge: $\log\snr$. Diversity gauge: $\log\snr$. Raw per-symbol
capacity coefficient: $\min(M,N)$; equivalently, atom-normalized
per-symbol DOF $\min(M,N)$, because the complex AWGN communication
atom has coefficient~$1$. The DMT coefficient $d^*(r)$ is the
operational load-$r$ reliability coefficient; at zero load,
$d^*(0)=MN$ is the endpoint DIV.

At the gauge level, finite-dimensional output-law packing and
covering are logarithmic in~$\snr$. The exact coefficient
$T\min(M,N)$ is the operational capacity coefficient supplied by
Telatar's capacity theorem; the frontier calculation below is used
for the zero-rate Bhattacharyya endpoint.

\emph{Assumption note.} The zero-rate Bhattacharyya calculation
below only requires $T\ge M$, so that full-rank pairwise differences
can be chosen. The positive-rate Zheng--Tse DMT curve is quoted under
its standard coherent block-fading assumptions and is used here only
as an operational benchmark. This subsection does not rederive the
positive-rate DMT.

\textbf{Frontier result.}
Since $H$ is revealed to the receiver, the channel-induced output law
for a codeword $X$ is the joint law
\[
    P_{Y,H|X}=P_H P_{Y|X,H}.
\]
Therefore the Bhattacharyya coefficient between two codewords is
\[
    \bar B(X_1,X_2)
    =
    \E_H
    \left[
        B(P_{Y|X_1,H},P_{Y|X_2,H})
    \right].
\]
Conditioned on~$H$, the channel is AWGN with known channel matrix,
so by Lemma~\ref{lem:bhatt_same_cov}:
$d_B(P_{Y|X_1,H},P_{Y|X_2,H})
= \frac{\snr}{4\ln 2}\|H(X_1-X_2)\|_F^2$.
The averaged Bhattacharyya coefficient is
$\bar{B} = \E_H[B(P_{Y|X_1,H},P_{Y|X_2,H})]
= \E_H[\exp(-\frac{\snr}{4}\|HD\|_F^2)]$
where $D = X_1-X_2$.
For i.i.d.\ Rayleigh~$H$ and $D$ of rank~$m$,
$\|HD\|_F^2$ is a sum of $Nm$ independent $\chi^2$
variables scaled by the squared singular values of~$D$. At
high~$\snr$, $\bar{B} \doteq \snr^{-mN}$, hence
$\bar{d}_B(D) \sim m N\log\snr$
(each of the $Nm$~Rayleigh output-law factors contributing
$\sim\log\snr$).
For any codebook with full-rank pairwise differences
($m = \mathrm{rank}(D) = \min(M,T)$, which requires $T \ge M$),
$\bar{d}_B \sim \min(M,T)\cdot N\log\snr$: the diversity gauge
is~$\log\snr$.

Conversely, for any admissible pair $X_1,X_2$, the difference
$D=X_1-X_2$ has rank at most~$M$ and
$\|D\|_F\le 2\sqrt{T}$. Hence its nonzero singular values are
bounded independently of~$\snr$, and Lemma~\ref{lem:rank_product}
gives
\[
    \bar{d}_B(D)
    =
    N\sum_{j=1}^{\operatorname{rank}(D)}
    \log\!\left(1+\frac{\snr\sigma_j^2(D)}{4}\right)
    \le MN\log\snr+O(1).
\]
A full-rank pair with singular values bounded away from zero achieves
$MN\log\snr+O(1)$. Therefore
\[
    \Delta_B^*(2;\snr)\sim MN\log\snr
\]
for $T\ge M$.

\textbf{Gauge recovery and status.}
Capacity gauge: $\log\snr$. Zero-rate diversity gauge:
$\log\snr$. For $T \ge M$, a full-rank pair gives
$\bar{d}_B(D) \sim MN\log\snr$, matching $d^*(0) = MN$. The
Rayleigh reliability atom has coefficient~$1$, so the raw diversity
coefficient $MN$ is also the endpoint DIV\@. The positive-rate
Zheng--Tse DMT coefficient curve is governed by outage geometry and
is quoted as an operational benchmark.

\begin{table}[t]
\centering
\caption{Coherent Rayleigh MIMO: two slices of the bilinear
map $(H,X) \mapsto HX$.}
\label{tab:coh_rayleigh}
\small
\begin{tabular}{@{}p{1.6cm}p{3.1cm}p{3.1cm}@{}}
\toprule
& Rate side & Pairwise diversity side \\
\midrule
Slice & $X \mapsto HX$, & $H \mapsto HD$, \\
      & conditioned on typical~$H$ & $D = X_1\!-\!X_2$ fixed \\
Rank  & $T\!\cdot\!\min(M,N)$ &
$N\!\cdot\!\mathrm{rank}(D)$ \\
Max rank & $T\!\cdot\!\min(M,N)$ &
$N\!\cdot\!\min(M,T)$ \\
Mechanism & Covering signal image &
Averaging over Rayleigh fades \\
Gauge & $\log\snr$ & $\log\snr$ \\
\bottomrule
\end{tabular}
\end{table}

Classification: \textbf{same-gauge}.

\subsection{Noncoherent Block Fading}\label{subsec:ch_block}

\textbf{Model.}
$H \in \C^{M\times N}$ with i.i.d.\
$\mathcal{CN}(0,1)$ entries, constant over $T$~slots (coherence
length~$T$), unknown to the receiver. Assume $N \ge M$ and $T \ge 2M$, so the active noncoherent dimension
is~$M$. Per block:
$Y = \sqrt{\snr}\,XH + Z$ with $X \in \C^{T\times M}$,
$\|X\|_F^2 \le T$. For scaled-unitary inputs
($X^\dagger X = (T/M)\,I_M$), the output depends on $X$ only
through its column space~\cite{ZhengTse2002}, so the signal space
is the Grassmannian $\mathrm{Gr}(M,T)$. The prelog
$M(T-M)/T$ and the principal-angle calculation below use this
Grassmannian structure.

\textbf{Intuition.}
Subspace packing on the Grassmannian replaces point packing in
Euclidean space. Over one $T$-slot block, the number of resolvable
subspaces scales as $\snr^{M(T-M)}$, so
$C_{\mathrm{block}}(\snr) \sim M(T-M)\log\snr$. Dividing by~$T$
gives the conventional per-physical-symbol capacity
$C_{\mathrm{sym}}(\snr) \sim \frac{M(T-M)}{T}\log\snr$. Random fading pulls the binary endpoint from the fixed-$H$
linear~$\snr$ scale to the logarithmic $\log\snr$ scale, as in
coherent MIMO.

\textbf{Operational benchmark.}
Capacity: $C(\snr) \sim M(T-M)/T\cdot\log\snr$ per
symbol~\cite{ZhengTse2002}. Noncoherent DMT:
$d^*(0) = MN$~\cite{ZhengTse2002Allerton}.
Capacity gauge: $\log\snr$. Diversity gauge: $\log\snr$.

\textbf{Frontier result.}
The output per block is $Y = \sqrt{\snr}\,XH + Z$, which has
distribution $\mathcal{CN}(0, \snr XX^\dagger + I_T)$ per receive
antenna (marginalizing over~$H$). By
Lemma~\ref{lem:bhatt_same_mean} applied to the $N$~independent
receive-antenna outputs:
\begin{equation}\label{eq:bhatt_block}
d_B(P_{Y|X_1},P_{Y|X_2}) = N\log
\frac{\det\!\big(\frac{\Sigma_1+\Sigma_2}{2}\big)}
{\sqrt{\det\Sigma_1\cdot\det\Sigma_2}},
\end{equation}
where $\Sigma_i = \snr X_i X_i^\dagger + I_T$.

For two scaled-unitary inputs $X_1, X_2 \in \C^{T\times M}$
satisfying $X_i^\dagger X_i = (T/M)\,I_M$, spanning
distinct $M$-dimensional subspaces with principal angles
$\theta_1,\dots,\theta_M$, the eigenvalues of $\Sigma_i$ are
$\snr\,(T/M)+1$ ($M$~times) and $1$ ($T-M$~times).
(Equivalently, one may write the calculation with unitary
representatives after replacing $\snr$ by $\snr T/M$; this changes
only $O(1)$ terms and does not affect the $\log\snr$ coefficient.)
With the original scaled-unitary normalization, the nontrivial
eigenvalues of $(\Sigma_1+\Sigma_2)/2$ are
$1+\frac{\snr T}{2M}(1\pm\cos\theta_k)$ for $k=1,\dots,M$
and $1$ ($T-2M$~times). At high~$\snr$, each principal-angle
pair contributes
$\log\!\big(1+\frac{(\snr T/M)^2\sin^2\theta_k}{4(\snr T/M+1)}\big)
\sim \log(\snr\sin^2\theta_k/4) + O(1)$
when $\sin\theta_k > 0$, giving
$d_B \sim N\sum_{k=1}^M \log\!\big(\snr\sin^2\theta_k/4\big)
= NM\log\snr + O(1)$.
For maximally separated subspaces ($\sin\theta_k = 1$ for
all~$k$, which requires $T \ge 2M$): $d_B \sim NM\log\snr$.

The packing number on the Grassmannian satisfies
\[
\log K_{\mathrm{pack}} \sim M(T-M)\log\snr.
\]
This is the coefficient-level Grassmannian instance of
Theorem~\ref{lem:sandwich}: fixed-resolution subspace packing and
covering have entropy coefficient $M(T-M)$, and the local covariance
KL radius is bounded at that resolution. Hence
\[
C_{\mathrm{block}}(\snr) \sim M(T-M)\log\snr,
\]
equivalently
$C_{\mathrm{sym}}(\snr) \sim \frac{M(T-M)}{T}\log\snr$.
The lower bound gives
$\Delta_B^*(2;\snr) \ge NM\log\snr + O(1)$.

\emph{Upper bound.} Let $A_i=X_iX_i^\dagger$, so
$\operatorname{rank}(A_i)\le M$ and $\operatorname{tr}(A_i)\le T$.
Write $\Sigma_i=I+\snr A_i$. Then
$(\Sigma_1+\Sigma_2)/2 = I+(\snr/2)(A_1+A_2)$.
By Lemma~\ref{lem:det_submult}, with $X=(\snr/2)A_1$ and
$Y=(\snr/2)A_2$,
\[
    \det\!\left(I+\frac{\snr}{2}(A_1+A_2)\right)
    \le
    \det\!\left(I+\frac{\snr}{2}A_1\right)
    \det\!\left(I+\frac{\snr}{2}A_2\right)
    \le
    \det(I+\snr A_1)\det(I+\snr A_2).
\]
Therefore
\begin{align*}
    \log
    \frac{
        \det((\Sigma_1+\Sigma_2)/2)
    }{
        \sqrt{\det\Sigma_1\det\Sigma_2}
    }
    &\le
    \frac{1}{2}\log\det\Sigma_1
    +
    \frac{1}{2}\log\det\Sigma_2  \\
    &\le
    M\log\snr+O(1),
\end{align*}
because each $A_i$ has rank at most~$M$ and trace bounded by~$T$.
Multiplying by the $N$~independent receive antennas gives
\[
    d_B(P_{Y|X_1},P_{Y|X_2})
    \le
    NM\log\snr+O(1).
\]
The maximally separated scaled-unitary construction above gives the
matching lower bound. Hence
\[
    \Delta_B^*(2;\snr)\sim NM\log\snr.
\]

\textbf{Gauge recovery and status.}
Capacity gauge: $\log\snr$. Zero-rate diversity gauge: $\log\snr$.
The binary endpoint gives
$\Delta_B^*(2;\snr) \sim NM\log\snr$, matching the zero-rate
noncoherent DMT coefficient~$MN$. The reliability atom has
coefficient~$1$, so the raw diversity coefficient $MN$ is also the
endpoint DIV\@. The load-$r$ frontier gauge is left for future
work.

Classification: \textbf{same-gauge}.

\subsection{Noncoherent Fast Fading}\label{subsec:ch_fast}

\textbf{Model.}
A single-antenna transmitter sends scalar $X_t$ per symbol slot;
$N$~receive antennas each see an independent fade:
\begin{equation}\label{eq:model_fast}
Y_t^{(j)} = \sqrt{\snr}\,H_t^{(j)} X_t + Z_t^{(j)},
\quad j=1,\dots,N,
\end{equation}
where $H_t^{(j)} \sim \mathcal{CN}(0,1)$ and
$Z_t^{(j)} \sim \mathcal{CN}(0,1)$ are mutually independent, drawn
independently at each symbol slot, unknown to the receiver.
Peak power: $|X_t| \le 1$.

\textbf{Intuition.}
Fading radializes the channel: the output depends on $x$ only
through $|x|^2$ (Lemma~\ref{lem:energy_only}). The output family is
a scale family $\{\mathcal{CN}(0,v): v=\snr|x|^2+1\}$ parameterized
by the variance $v\in[1,1+\snr]$. In log-variance coordinates
$u = \ln v \in [0,L]$ with $L = \ln(1+\snr)\sim\ln\snr$, the
Bhattacharyya distance depends only on the spacing
$|u_1-u_2|$. The packing space is a one-dimensional interval of
length $\sim\ln\snr$, so the number of distinguishable points is
$\sim\log\snr$, giving capacity gauge $\log\log\snr$.
(Section~\ref{subsec:int_fast}).

\textbf{Operational benchmark.}
Capacity: $C(\snr)\sim\log\log\snr$~\cite{LapidothMoser2003}.
Capacity gauge: $\log\log\snr$. Raw capacity coefficient on the
$\log\log\snr$ gauge: $1$, matching the scalar fast-fading
communication atom. Zero-rate diversity gauge: identified below by
the binary endpoint.

\textbf{Frontier result.}

\begin{lemma}[Energy-only reduction]\label{lem:energy_only}
Each output law depends on $x$ only through $S = |x|^2$:
$P_{Y^{(j)}|X=x}^{(\snr)} = \mathcal{CN}(0,\snr|x|^2+1)$.
The $N$ outputs are conditionally independent given~$X$.
\end{lemma}
\begin{proof}
Marginalizing over $H^{(j)} \sim \mathcal{CN}(0,1)$:
$Y^{(j)} = \sqrt{\snr}\,H^{(j)} x + Z^{(j)}$
has distribution $\mathcal{CN}(0, \snr|x|^2 + 1)$, depending
on~$x$ only through $|x|^2$. Independence across~$j$ follows from
the independence of $\{H^{(j)}, Z^{(j)}\}$.
\end{proof}

\begin{lemma}[Bhattacharyya distance for scale families]
\label{lem:scale_bhatt}
The per-antenna Bhattacharyya distance between $P_{v_1}$ and
$P_{v_2}$ is (in base-2 logarithms)
\begin{equation}\label{eq:bhatt_scale}
d_B^{\mathrm{single}}(v_1,v_2) =
\log_2\cosh\!\big(\tfrac{\ln v_1 - \ln v_2}{2}\big).
\end{equation}
\end{lemma}
\begin{proof}
By Lemma~\ref{lem:bhatt_same_mean} with $\Sigma_i = v_i$ (scalar):
$B = 2\sqrt{v_1 v_2}/(v_1+v_2)
= 1/\cosh(\Delta/2)$,
where $\Delta = |\ln v_1 - \ln v_2|$, giving
$d_B = \log\cosh(\Delta/2)$.
\end{proof}

By multiplicativity, $N$~receive antennas give total Bhattacharyya
distance $d_B = N \cdot d_B^{\mathrm{single}}$.

\begin{lemma}[Gauge class under peak vs.\ average power]
\label{lem:constraint}
Under both peak power ($|X| \le 1$) and average power
($\E[|X|^2] \le 1$), the packing complexity satisfies
$K_{\mathrm{pack}}(\delta;\snr) = \log\log\snr + O(1)$: gauge
class $\log\log\snr$.
\end{lemma}
\begin{proof}
\emph{Peak power.}
The log-variance coordinate $u\in[0,L]$ with $L=\ln(1+\snr)$.
Codebooks with Bhattacharyya separation $\ge\delta$ require
$u$-coordinates that are $c(\delta)$-separated, so
$N_{\mathrm{pack}} \le 1+L/c(\delta)$. Equally spaced points
achieve the matching lower bound:
$K_{\mathrm{pack}} = \log\log\snr + O(1)$.

\emph{Average power.}
$\setP^{(\snr)}_{\mathrm{avg}} \supseteq
\setP^{(\snr)}_{\mathrm{peak}}$ gives
$K_{\mathrm{pack}}^{\mathrm{avg}} \ge \log\log\snr + O(1)$.
For the upper bound, consider a size-$K$ codebook
$\{x_1,\dots,x_K\}$ with average power
$(1/K)\sum_{i=1}^{K}|x_i|^2 \le 1$. Each codeword satisfies
$|x_i|^2 \le K$ (otherwise a single codeword would exhaust the
average-power budget), so the log-variance coordinates satisfy
$u_i = \ln(\snr|x_i|^2+1) \le \ln(1+K\snr)$.
The Bhattacharyya separation $\ge\delta$ requires
$c(\delta)$-spacing, so
$K \le 1 + \ln(1+K\snr)/c(\delta)$.
Since $\ln(1+K\snr) \le \ln K + \ln\snr + O(1)$, this gives
$K \le C_\delta\,\ln(K\snr)$, which implies
$K = O(\log\snr)$ (the self-bounding inequality
$K \le C_\delta\ln(K\snr)$ forces $K$ to grow at most
logarithmically). Hence
$K_{\mathrm{pack}}^{\mathrm{avg}} \le \log\log\snr + O(1)$.
\end{proof}

The capacity-gauge statement is also an instance of
Theorem~\ref{lem:sandwich}. Under peak power, the output-law family is
the scale family
\[
P_u = \mathcal{CN}(0, e^u),\qquad
u = \ln(1+\snr|x|^2) \in [0, L_\snr],
\]
where
\[
L_\snr = \ln(1+\snr) \sim \ln\snr.
\]
For $N$ receive antennas,
\[
B(P_u^{\otimes N}, P_{u'}^{\otimes N})
= \cosh\!\Big(\frac{|u-u'|}{2}\Big)^{-N}.
\]
Choose a fixed spacing $a > 0$ and place grid points $u_i = ia$ in
$[0, L_\snr]$. Then
\[
M_\snr \asymp L_\snr \asymp \log\snr,\qquad
\log M_\snr = \log\log\snr + O(1).
\]
The aggregate Bhattacharyya overlap is bounded uniformly in~$\snr$:
\[
\sup_i \sum_j B(P_{u_i}^{\otimes N}, P_{u_j}^{\otimes N})
\le 1 + 2\sum_{k=1}^{\infty}
\cosh\!\Big(\frac{ka}{2}\Big)^{-N} < \infty.
\]
Conversely, cover $[0, L_\snr]$ by intervals of fixed length~$a$.
The number of intervals is $\asymp\log\snr$. If $u, u'$ lie in the
same interval, then
\[
D(P_u^{\otimes N} \| P_{u'}^{\otimes N})
= \frac{N(u' - u + e^{u-u'} - 1)}{\ln 2}
\]
is bounded by a constant depending only on~$a$ and~$N$.
Theorem~\ref{lem:sandwich} therefore gives the coefficient-level
capacity statement
\[
C(\snr) = \log\log\snr + O(1).
\]
Thus fixed-separation packing of the radialized scale family recovers
both the capacity gauge $\log\log\snr$ and the raw capacity
coefficient~$1$. Since the scalar fast-fading communication atom has
coefficient~$1$ on the $\log\log\snr$ gauge, the endpoint DOF
is~$1$. A self-contained structural version of the scale-family
packing argument is given in Proposition~\ref{prop:scale_family}.

\begin{theorem}[Exact fast-fading frontier and endpoint]
\label{thm:cross}
For the fast-fading model with $N$~receive antennas, the single-use
Bhattacharyya frontier at
\[
    K=\lceil(\log\snr)^r\rceil,\qquad 0<r<1,
\]
satisfies
\begin{equation}\label{eq:cross_tradeoff}
\Delta_B^*\big(\lceil(\log\snr)^r\rceil;\,\snr\big)
\;\sim\; \tfrac{N}{2}\,(\log\snr)^{1-r}.
\end{equation}
Thus the load-$r$ frontier gauge is exactly
$(\log\snr)^{1-r}$.
At the binary endpoint $K=2$,
\begin{equation}\label{eq:fast_binary}
\Delta_B^*(2;\snr) \;\sim\; \tfrac{N}{2}\,\log\snr.
\end{equation}
At $r=1$: $\Delta_B^* \to N\log_2\!\cosh(\ln 2/2)$.
Consequently, for every $\varepsilon>0$ and all sufficiently
large~$\snr$, $n$~independent repetitions of the same single-use codeword
give a repetition code satisfying
\[
P_e \le
\big(\lceil(\log\snr)^r\rceil-1\big)\,
2^{-n(N/2-\varepsilon)(\log\snr)^{1-r}}.
\]
Equivalently,
$P_e \le 2^{r\log\log\snr - n(N/2+o(1))(\log\snr)^{1-r}}$.
\end{theorem}

\begin{proof}
\emph{Achievability.} Place $K = \lceil(\log\snr)^r\rceil$ equally
spaced log-variance levels in $[0,L]$. The spacing is
$\Delta_u = L/(K-1) = (\ln 2)(\log\snr)^{1-r}(1+o(1))$.
For $r < 1$, $\Delta_u \to \infty$; using the large-argument
expansion $\log\cosh(\Delta/2) = \Delta/(2\ln 2) - 1 + O(e^{-\Delta})$:
$d_B^{\mathrm{single}}(\Delta_u) = \tfrac{1}{2}(\log\snr)^{1-r}
(1+o(1))$. Total over $N$~antennas:
$\Delta_{\min}^{\mathrm{ach}}
= (N/2)(\log\snr)^{1-r}(1+o(1))$.
Substituting this $\Delta_{\min}$ into
Corollary~\ref{cor:optimal_packing} gives the stated
$\varepsilon$-form of the repetition-code reliability bound.

\emph{Converse.} Pigeonhole: any $K$ points in $[0,L]$ have minimum
spacing $\le L/(K-1)$. Since~\eqref{eq:bhatt_scale} is strictly
increasing in $|u_1-u_2|$, $\Delta_B^*(K;\snr) \le
N\log\cosh(L/(2(K-1))) = (N/2)(\log\snr)^{1-r}(1+o(1))$.
\end{proof}

\begin{corollary}[Fast-fading endpoint DIV]
\label{cor:fast_div}
By~\eqref{eq:fast_binary}, the binary endpoint has raw coefficient
$N/2$ on the $\log\snr$ gauge.
The scalar fast-fading reliability atom has
coefficient $1/2$ on this gauge. Therefore the
endpoint DIV is
\[
  \mathrm{DIV}_{B,0} = \frac{N/2}{1/2} = N.
\]
For $0 < r < 1$, Theorem~\ref{thm:cross} gives the load-$r$
frontier $(N/2)(\log\snr)^{1-r}$.
\end{corollary}

Together with the Bhattacharyya
sandwich~(Proposition~\ref{thm:bridge}), this identifies the load-$r$ frontier gauge for
$K=\lceil(\log\snr)^r\rceil$, and the
corresponding repetition-code gauge, as $(\log\snr)^{1-r}$.
The repetition code has $K$ messages over $n$~uses, so its block rate
is $(\log K)/n \to 0$; the load parameter~$r$ governs the single-use
codebook size, not the $n$-block rate. The fully optimized Shannon
reliability function at positive block rate, optimized over arbitrary
$n$-block codes, is a separate operational object. At
the binary endpoint, the raw diversity coefficient is~$N/2$, and the
endpoint DIV is~$N$. For $0 < r < 1$, the same numerical coefficient
$N/2$ appears as the load-$r$ coefficient on the gauge
$(\log\snr)^{1-r}$.

\textbf{Gauge recovery.}
Capacity gauge: $\log\log\snr$ (exact, recovers
Lapidoth--Moser~\cite{LapidothMoser2003}).
Zero-rate diversity gauge: $\Delta_B^*(2;\snr) \sim (N/2)\log\snr$,
so the zero-rate diversity gauge is $\log\snr$---a gauge identified
here by the binary endpoint.
Raw capacity coefficient on the $\log\log\snr$ gauge: $1$. Raw
diversity coefficient: $N/2$. Dividing by the reliability-atom
coefficient $1/2$ gives endpoint DIV~$N$.

\textbf{Status:} exact load-$r$ frontier for $0 < r < 1$ and exact
binary endpoint. Zero-rate diversity gauge $\log\snr$; endpoint
DIV~$N$. Classification: \textbf{cross-gauge}.

\textbf{Interpretation.}
The cross-gauge phenomenon arises because the packing space (a
log-variance interval of length $\sim\ln\snr$) is vastly larger than
the rate budget ($\sim\log\log\snr$). The classical prelog
$C(\snr)/\log\snr \to 0$ because capacity grows as
$\log\log\snr$~\cite{LapidothMoser2003}, confirming the $\log\snr$
gauge is mismatched on the capacity side. On the correct diversity
gauge $\log\snr$, the raw diversity coefficient is $N/2$. The factor $N$ arises from the $N$~independent output-law factors,
i.e.\ $N$ unit-normalized output reliability looks after scale-family
normalization. The factor $1/2$ is the fast-fading one-look
reliability-atom coefficient removed by atomic normalization.

\subsection{Multipath Fading}\label{subsec:ch_multipath}

\textbf{Model.}
An $L_p$-path extension of fast fading with $N$~receive antennas:
$Y_{j,t} = \sqrt{\snr}\sum_{\ell=0}^{L_p-1}
H_{j,\ell,t}\,X_{t-\ell} + Z_{j,t}$, where
$H_{j,\ell,t}\sim\mathcal{CN}(0,\sigma_\ell^2)$ independently across
$(j,\ell,t)$, and $Z_{j,t}\sim\mathcal{CN}(0,1)$. Let
$L_p^+ = |\{\ell : \sigma_\ell^2 > 0\}|$.
The calculation below uses the guarded subcode consisting of one data
symbol~$x$, $|x| \le 1$, followed by $L_p-1$ zero guard symbols. One
guarded block occupies $L_p$~physical symbol slots. Coefficients are
reported per effective guarded data symbol, with the per-physical-symbol
version obtained by dividing by~$L_p$.

\textbf{Intuition and operational benchmark.}
Koch--Lapidoth capacity~\cite{KochLapidoth2009} gives the capacity
gauge $\log\log\snr$. The guarded construction below separates data
symbols by $L_p-1$ guard symbols, reducing each active data symbol to
independent fast-fading-like scale-family observations along the
nonzero paths. The result is a lower bound per effective guarded data
symbol.

Let $\Delta_{B,\mathrm{guard}}^*(K;\snr)$ denote the Bhattacharyya
frontier restricted to this guarded subcode.

\textbf{Frontier result.}
With guard intervals of length $L_p-1$ between data symbols,
each data symbol~$X_t$ produces $L_p$ delayed outputs per receive
antenna:
\[
Y_{j,t+\ell} = \sqrt{\snr}\,H_{j,\ell,t+\ell}\,X_t + Z_{j,t+\ell},
\qquad \ell = 0,\dots,L_p-1.
\]
Of these $L_p$ physical path positions, only the
\[
L_p^+ = |\{\ell : \sigma_\ell^2 > 0\}|
\]
nonzero-variance paths contribute a high-SNR scale-family term. If
$\sigma_\ell^2 = 0$, then the corresponding output distribution is
independent of~$x$ and contributes zero Bhattacharyya distance. For
$\sigma_\ell^2 > 0$,
\[
Y_{j,t+\ell} \mid X_t = x
\sim \mathcal{CN}(0,\,1+\snr\sigma_\ell^2|x|^2).
\]
Therefore, by multiplicativity, the total guarded Bhattacharyya
distance is
\[
d_B = \sum_{j=1}^{N}\sum_{\ell:\,\sigma_\ell^2>0}
d_B^{\mathrm{single}}(1+\snr\sigma_\ell^2|x_1|^2,\,
1+\snr\sigma_\ell^2|x_2|^2).
\]
Thus the physical delayed-output count is $NL_p$, but the number of
$\log\snr$-scale unit-normalized output reliability contributions is
$NL_p^+$. Zero-variance paths produce physical delayed outputs but no
high-SNR output reliability looks.
Since each path is a scale family with log-variance range
$\ln(1 + \snr\sigma_\ell^2) \sim \ln\snr$ (for
$\sigma_\ell^2 > 0$), the packing analysis from
Theorem~\ref{thm:cross} applies: the input $x$ maps to
log-variance coordinates that are all monotone functions of
$|x|^2$, so the packing space is one-dimensional and
$K_{\mathrm{pack}} = \log\log\snr + O(1)$ (same gauge as
fast fading).

This is a guarded-construction lower bound measured per effective
guarded data symbol. One data symbol is separated from the next by
$L_p-1$ guard symbols, so the coefficient displayed below is not a
per-physical-symbol reliability coefficient. The result proves an available scale for this construction only; it
does not identify the exact unguarded frontier or the exact
per-physical-symbol diversity coefficient.

For the restricted guarded binary endpoint,
\[
  \Delta_{B,\mathrm{guard}}^*(2;\snr)
  \sim \frac{NL_p^+}{2}\,\log\snr.
\]
Since the guarded subcode is admissible for the full channel,
\[
  \Delta_B^*(2;\snr) \ge \Delta_{B,\mathrm{guard}}^*(2;\snr).
\]
This is a lower bound for the guarded subcode; no matching converse
for the full unguarded frontier is proved. Per physical
symbol slot under the guarded normalization, the displayed coefficient
is divided by~$L_p$.

\textbf{Gauge recovery.}
Capacity gauge: $\log\log\snr$, matching
Koch--Lapidoth~\cite{KochLapidoth2009}. The guard-interval
construction gives the zero-rate lower bound
$\Delta_B^*(2;\snr) \ge \frac{NL_p^+}{2}\log\snr$,
where $L_p^+$ is the number of nonzero-variance paths. The lower bound proves cross-gauge behavior, since the diversity
side is at least on the $\log\snr$ scale while capacity grows on
the $\log\log\snr$ scale. After reliability-atom normalization, the
guarded endpoint DIV lower bound is $NL_p^+$ per
effective guarded data symbol. A matching upper bound, and hence the
exact zero-rate diversity gauge and coefficient, are not determined
here.

\textbf{Status:} capacity gauge exact; guarded zero-rate lower bound;
guarded endpoint DIV lower bound $NL_p^+$ per effective
guarded data symbol, or $NL_p^+/L_p$ per physical symbol under the
guarded normalization. Exact unguarded frontier open. Classification:
\textbf{cross-gauge} lower bound.

\subsection{Parallel Fractional-Log Fading: Conditional Gallager
Benchmark}\label{subsec:ch_parallel}

This section is included only as a boundary-case benchmark. The
stationary Toeplitz model has the fractional-log capacity gauge but
does not diagonalize into independent scalar subchannels for pairwise
covariance comparisons. The parallel model removes that obstruction
by assumption, so it shows what a factorized Gallager lower bound
looks like on the fractional-log total gauge. It is not used as a
proof of the stationary Toeplitz diversity gauge.

\textbf{Model.}
A parallel channel with $J$~subchannels and $N$~receive antennas
per subchannel:
\begin{equation}\label{eq:model_parallel}
Y_k^{(j)} = \sqrt{\snr}\,H_k^{(j)}\,X_k + Z_k^{(j)},
\quad k=1,\dots,J,\;\; j=1,\dots,N,
\end{equation}
where $H_k^{(j)} \sim \mathcal{CN}(0,\lambda_k)$ are independent
across both $k$ and~$j$, known to the receiver, and
$Z_k^{(j)} \sim \mathcal{CN}(0,1)$. The spectral profile
$\{\lambda_k\}_{k=1}^{J}$ is a \emph{model-defining assumption}:
we require that $J = J(\snr)$ and the eigenvalue profile
$\{\lambda_k\}$ are jointly chosen so that
\begin{equation}\label{eq:spectral_condition}
\sum_{k=1}^{J}\log(1+\snr\lambda_k) \;\asymp\;
J\,(\log\snr)^\beta, \qquad \beta \in (0,1).
\end{equation}
(For fixed finite~$J$ with fixed $\lambda_k > 0$, the sum scales
as $J\log\snr$, not $J(\log\snr)^\beta$; the fractional-log scaling
requires vanishing eigenvalues, as in the Toeplitz model of
Section~\ref{subsec:ch_toeplitz}, where the spectral cusp produces
$\beta \in (0,1)$.)

Fix a product input ensemble
$P_X^{(\snr)} = \prod_{k=1}^{J} P_{X_k}^{(\snr)}$ satisfying
the stated power constraint. The Gallager function below is
computed for this chosen ensemble. No optimization over all input
distributions is claimed unless explicitly stated.

\textbf{Intuition.}
The subchannels are independent, so the Gallager $E_0$ function
factors coordinatewise. This makes the conditional lower-bound
exponent tractable on the total gauge
$G_J(\snr) = J(\snr)(\log\snr)^\beta$, the same total gauge used
for capacity in this benchmark.

\textbf{Operational benchmark.}
Capacity:
$C(\snr) = \sum_k \E\!\big[\log\!\big(1+\snr\sum_{j=1}^{N}
|H_k^{(j)}|^2\big)\big] \asymp J(\log\snr)^\beta$.
The quantity $G_J(\snr):=J(\snr)(\log\snr)^\beta$ is the total block
gauge for this parallel benchmark; the normalized per-subchannel gauge
is $(\log\snr)^\beta$.
Proposition~\ref{thm:gallager_parallel} is stated on the total
gauge $G_J(\snr)$. Table~\ref{tab:audit} reports the total gauge
unless explicitly marked ``normalized.''

\textbf{Conditional Gallager exponent result.}

\begin{proposition}[Conditional Gallager lower-bound scaling for the
parallel fractional-log model]\label{thm:gallager_parallel}
For a spectral profile satisfying~\eqref{eq:spectral_condition} and
a product input ensemble
$P_X^{(\snr)}=\prod_{k=1}^J P_{X_k}^{(\snr)}$, suppose that for
each $s$ in a nonempty set $S\subset(0,1]$ the normalized liminf
\[
\underline{c}_{P_X}(s)
    :=
    \liminf_{\snr\to\infty}
    \frac{E_0(s;\,P_X^{(\snr)})}{G_J(\snr)}
\]
is positive. Then at total rate $R = r\,c_1\,G_J(\snr)$ the
random-coding exponent for this ensemble satisfies
\begin{equation}\label{eq:gallager_parallel}
\liminf_{\snr\to\infty}
    \frac{E_r(R;\snr)}{G_J(\snr)}
    \;\ge\;
    \sup_{s\in S}
    \{\underline{c}_{P_X}(s) - s\,r\,c_1\}.
\end{equation}
When the right-hand side is positive, this gives an achievable
random-coding error exponent on the $G_J(\snr)$ scale.

Here $c_1$ is used as an exact capacity coefficient when an exact
statement $C(\snr) \sim c_1 G_J(\snr)$ is available. If the available
statement is only $C(\snr) \asymp G_J(\snr)$, then $c_1$ is a fixed
benchmark normalization constant, not a canonical capacity
coefficient.

This is conditional and does not determine the optimal zero-rate
diversity gauge.
\end{proposition}

\begin{proof}
The total Gallager function is $E_0(s) = \sum_{k=1}^{J} E_0^{(k)}(s)$,
where $E_0^{(k)}(s)$ is the per-subchannel Gallager function for
subchannel~$k$ with variance $\lambda_k$ and $N$~receive antennas.
The hypothesis $\underline{c}_{P_X}(s) > 0$ for $s\in S$
and the spectral condition~\eqref{eq:spectral_condition}
give $\liminf E_0(s)/G_J(\snr) \ge \underline{c}_{P_X}(s) > 0$.
The Gallager exponent $E_r(R) = \max_{0\le s\le 1}[E_0(s)-sR]$
with $R = r\,c_1\,G_J(\snr)$
gives~\eqref{eq:gallager_parallel}.
\end{proof}

\textbf{Status:} conditional lower-bound benchmark; endpoint DIV
undetermined. Classification: \textbf{conditional lower-bound
benchmark}.

\begin{remark}[Parallel benchmark versus Toeplitz]
\label{rem:parallel_toeplitz}
The parallel model is not an exact diagonalization of the Toeplitz
problem: in the Toeplitz model, the output covariances
$I+\snr D_X R_T D_X^\dagger$ do not generally commute across
codewords, obstructing coordinatewise factorization.
\end{remark}

\subsection{Stationary Toeplitz Fading}\label{subsec:ch_toeplitz}

\textbf{Model.}
Scalar channel with $N$~receive antennas:
$Y_t^{(j)} = \sqrt{\snr}\,H_t^{(j)}\,X_t + Z_t^{(j)}$, where
$\{H_t^{(j)}\}_{t \ge 1}$ is a stationary Gaussian process with
power spectral density (PSD) $f_H(\lambda)$ on $[-\pi,\pi]$, independent across the
$N$~receive antennas. The fading process is not revealed to the
receiver; the receiver knows its law. The PSD is non-regular: for example, Lapidoth's cusp family
$f_H(\lambda) \asymp \exp(-c/|\lambda|^\alpha)$ with $\alpha > 1$
yields a fractional-log exponent
$\beta = (\alpha-1)/\alpha \in (0,1)$.
Equivalently, writing $\alpha = \gamma+1$ gives
$\beta = \gamma/(\gamma+1)$.

\textbf{Intuition.}
Lapidoth~\cite{Lapidoth2005} showed that certain non-regular spectra
give $C(\snr) \asymp (\log\snr)^\beta$ with
$\beta \in (0,1)$. The capacity gauge is
$(\log\snr)^\beta$.

\textbf{Operational benchmark.}
Capacity: $C(\snr) \asymp (\log\snr)^\beta$~\cite{Lapidoth2005}.
Capacity gauge: $(\log\snr)^\beta$.
Zero-rate diversity gauge: not determined by the cited capacity
result and left open here.

\textbf{Frontier observation.}
Under a per-symbol peak-power constraint $|X_t|^2 \le P$,
consider the on--off pair $x = \sqrt{P}\,\mathbf{1}_T$ and
$x = \mathbf{0}$. Per receive antenna,
$Y_T \mid x=\mathbf{0} \sim \mathcal{CN}(0,I_T)$ and
$Y_T \mid x=\sqrt{P}\,\mathbf{1}_T \sim
\mathcal{CN}(0,\,I_T + \snr P\,R_T)$,
where $R_T$ is the $T\times T$ Toeplitz fading covariance matrix.
Per receive antenna, Lemma~\ref{lem:bhatt_same_mean} gives
\[
d_{B,1}^{(\text{on-off})}
= \log\det\!\Big(I_T + \tfrac{\snr P}{2}\,R_T\Big)
  - \tfrac{1}{2}\log\det\!\big(I_T + \snr P\,R_T\big).
\]
With $N$ independent receive antennas,
\[
d_{B,N}^{(\text{on-off})}
= N\!\left[\log\det\!\Big(I_T + \tfrac{\snr P}{2}\,R_T\Big)
  - \tfrac{1}{2}\log\det\!\big(I_T + \snr P\,R_T\big)\right].
\]
This identity is only a two-point candidate lower bound. It does
not identify the stationary Toeplitz diversity gauge. In particular,
one cannot infer a $(\log\snr)^\beta$ diversity lower bound by
applying a Toeplitz log-det limit with the scalar multiplier held
fixed and then sending that multiplier to infinity. The joint high-SNR/long-block nonregular regime
underlying Lapidoth's fractional-log capacity result is different.
A matching upper bound on $\Delta_B^*(2;\snr)$, or any exact lower
bound at the correct stationary diversity scale, remains open.

\textbf{Load-$r$ frontier gauge: open.}
The load-$r$ frontier gauge is open for the same reason: the output covariance has the
form
\[
    I+\snr D_X R_T D_X^\dagger,
\]
and these matrices do not generally commute across different
codewords. Consequently the Gallager $E_0$ does not factor
coordinatewise, unlike in the parallel benchmark, and the eigenvalue
structure depends on pairs of inputs jointly.

\textbf{Status: open.} The preceding on-off calculation does not
establish a $(\log\snr)^\beta$ diversity gauge. Classification:
\textbf{open}.

\section{Gauge Audit}\label{sec:audit}

Table~\ref{tab:audit} is the main gauge audit: a channel-by-channel
comparison of capacity gauges, zero-rate diversity gauges, endpoint
DOF/DIV, and channel classification. Table~\ref{tab:frontier_audit}
records frontier results, load-$r$ status, and proof status.

\begin{table*}[ht]
\centering
\caption{Gauge audit: capacity gauge, zero-rate diversity gauge,
endpoint DOF/DIV, and classification. Fixed-$H$ entries are
block-level unless divided by~$T$.
``Frac-log'' abbreviates ``fractional-log.''}
\label{tab:audit}
\footnotesize
\begin{tabular}{@{}p{2.8cm}lp{2.8cm}p{4.8cm}p{2.0cm}@{}}
\toprule
Channel & Capacity gauge & Zero-rate diversity gauge
& Endpoint DOF/DIV & Classification \\
\midrule
Fixed-$H$
& $\log\snr$ & $\snr$
& block endpoint DOF $Tr_H$; endpoint DIV
$T\sigma_1^2(H)$, gain-dependent for unnormalized $H$
& \textbf{cross-gauge} \\
Coherent Rayleigh MIMO
& $\log\snr$ & $\log\snr$
& per-symbol DOF $\min(M,N)$; block DOF $T\min(M,N)$; zero-rate
DIV $MN$
& same-gauge \\
Noncoherent block fading
& $\log\snr$ & $\log\snr$
& block DOF $M(T\!-\!M)$; per-symbol DOF $M(T\!-\!M)/T$; zero-rate
DIV $MN$
& same-gauge \\
Noncoherent fast fading
& $\log\log\snr$ & $\log\snr$
& DOF $1$; endpoint DIV $N$
& \textbf{cross-gauge} \\
Multipath fast fading
& $\log\log\snr$ & at least $\log\snr$ by guarded subcode; exact
unguarded gauge not identified
& guarded endpoint DIV lower bound $NL_p^+$ per effective
guarded data symbol; $NL_p^+/L_p$ per physical symbol under guarded
normalization
& \textbf{cross-gauge} lower bound \\
Parallel frac-log
& $G_J(\snr)\!=\!J(\snr)(\log\snr)^\beta$ total
& not identified; conditional Gallager lower-bound benchmark on
$G_J(\snr)$
& endpoint DIV undetermined
& cond.\ lower-bound benchmark \\
Stationary Toeplitz frac-log
& $(\log\snr)^\beta$ & open
& endpoint DOF/DIV open
& open \\
\bottomrule
\end{tabular}

\smallskip
For unnormalized fixed deterministic~$H$, endpoint DIV
$T\sigma_1^2(H)$ is gain-dependent.

The diversity-side look is normalized; it is not a bare physical
coordinate, receive antenna, or singular direction.
\end{table*}

\begin{table*}[ht]
\centering
\caption{Frontier status audit: frontier results, load-$r$ status,
and proof status.
``Frac-log'' abbreviates ``fractional-log.''}
\label{tab:frontier_audit}
\footnotesize
\begin{tabular}{@{}p{2.8cm}p{4.5cm}p{3.8cm}p{4.5cm}@{}}
\toprule
Channel & Frontier result
& Load-$r$ status & Proof status \\
\midrule
Fixed-$H$
& $C_{\mathrm{block}}\!\sim\! Tr_H\log\snr$;
$\Delta_B^*(2;\snr)\!=\!\frac{T\sigma_1^2(H)}{\ln 2}\snr$
& load-$r$ frontier predicted at gauge level
& zero-rate exact; load-$r$ frontier left for future work \\
Coherent Rayleigh MIMO
& per-symbol capacity coefficient $\min(M,N)$; binary endpoint
${\sim}\,MN\log\snr$ for $T\!\ge\! M$
& positive-rate DMT governed by outage geometry
& zero-rate coefficient exact; positive-rate DMT quoted as
operational benchmark \\
Noncoherent block fading
& $C_{\mathrm{block}}\!\sim\! M(T\!-\!M)\log\snr$;
$\Delta_B^*(2;\snr)\!\sim\! MN\log\snr$
& load-$r$ frontier gauge left for future work
& zero-rate exact under stated assumptions \\
Noncoherent fast fading
& $\Delta_B^*(2;\snr)\!\sim\!(N/2)\log\snr$;
$\Delta_B^*(\lceil(\log\snr)^r\rceil;\snr)\!\sim\!(N/2)(\log\snr)^{1-r}$
& exact load-$r$ frontier for
$0\!<\!r\!<\!1$
& exact load-$r$ frontier for $0<r<1$; zero-rate exact;
repetition-code gauge identified \\
Multipath fast fading
& guarded lower bound
$\Delta_B^*(2;\snr)\!\ge\!(NL_p^+/2)\log\snr$ per effective
guarded data symbol
& guarded construction suggests load-$r$ lower bounds
& lower-bound only \\
Parallel frac-log
& conditional Gallager benchmark on $G_J(\snr)$
& no exact frontier theorem
& conditional; specified product ensemble and $E_0$-scaling
hypothesis \\
Stationary Toeplitz frac-log
& on-off log-det calculation is a two-point candidate
& zero-rate diversity gauge and load-$r$ frontier gauge open
& open \\
\bottomrule
\end{tabular}

\smallskip
At the fast-fading binary endpoint, $N/2$ is the raw diversity
coefficient; the corresponding endpoint DIV is~$N$ after division by
the fast-fading reliability atom~$1/2$. For $0 < r < 1$, $N/2$ is
the load-$r$ coefficient on the gauge $(\log\snr)^{1-r}$.

The fixed-$H$ binary endpoint uses the total-power output-look
convention from Table~\ref{tab:audit}; a per-mode power convention
would define a different all-mode coefficient, but that is not the
frontier convention used here.
\end{table*}

Tables~\ref{tab:audit} and~\ref{tab:frontier_audit} separate endpoint
interpretation from proof status. Fixed-$H$ and fast fading are the
main cross-gauge examples. Coherent Rayleigh MIMO and noncoherent
block fading are same-gauge calibration cases: both have $\log\snr$
capacity and zero-rate diversity gauges, with zero-rate DIV~$MN$
under the stated assumptions. Multipath, parallel fractional-log, and
stationary Toeplitz entries are deliberately classified as
lower-bound, conditional, or open.

\section{Conclusion}\label{sec:conclusion}

This paper separates two operations that are merged in the classical
coherent-MIMO convention: selecting the correct high-SNR gauge and
normalizing the resulting coefficient by the appropriate communication
or reliability atom. In coherent complex MIMO these steps are
invisible because both capacity and diversity live on the $\log\snr$
gauge and the relevant atom coefficients are one. Outside that
calibration case, the separation is essential.

The frontier provides the common geometric tool. Fixed-separation
packing gives the rate-side output-law entropy, and the
capacity--packing sandwich turns this entropy into the capacity gauge
when the fixed-resolution packing and covering gauges match. The
binary endpoint identifies the zero-rate diversity endpoint. This
reproduces the classical same-gauge calibration cases: coherent
Rayleigh MIMO and noncoherent block fading have $\log\snr$ capacity
and zero-rate diversity gauges, with zero-rate DIV~$MN$ under the
stated assumptions.

The main cross-gauge examples are fixed deterministic~$H$ and
noncoherent fast fading. For fixed~$H$, capacity is governed by
covering the image of $X \mapsto HX$, giving the $\log\snr$ capacity
gauge and block DOF $T\,\mathrm{rank}(H)$, whereas binary reliability is governed by pairwise output expansion,
giving the $\snr$ diversity gauge and endpoint DIV $T\sigma_1^2(H)$,
the optimized output-look count under the total block peak-power
constraint. For unnormalized~$H$, this endpoint DIV is
gain-dependent. For noncoherent fast fading,
radialization reduces the channel to a scale family: capacity grows on
the $\log\log\snr$ gauge with DOF~$1$, while the binary endpoint
grows on the $\log\snr$ zero-rate diversity gauge and gives endpoint DIV~$N$. At
positive load, the fast-fading load-$r$ frontier is exact and grows
on the gauge $(\log\snr)^{1-r}$.

The audit tables also mark the boundary of the present results.
Multipath fading is handled through a guarded construction, giving a
lower bound but not the exact unguarded frontier. The parallel
fractional-log model gives a conditional Gallager benchmark under an
explicit product-ensemble and $E_0$-scaling hypothesis, but it is not
an exact frontier theorem. The stationary Toeplitz model remains open
on the diversity side because the covariance matrices generated by
different codewords do not generally commute, preventing the
coordinatewise factorization available in the parallel benchmark.

\textbf{Open problems.}
Several open problems remain.

\emph{First}, the load-$r$ frontier should be determined beyond the
fast-fading model. The fixed-$H$ analysis gives exact endpoint gauges
and a gauge-level positive-load prediction, but matching upper and
lower bounds for the full load-$r$ frontier are not proved. Coherent
Rayleigh MIMO and noncoherent block fading also have exact zero-rate
endpoints, but their load-$r$ frontiers are not
identified in this paper. A complete theory would determine when the
frontier gauge interpolates cleanly between the capacity endpoint and
the binary endpoint, and when new geometry appears at positive load.

\emph{Second}, the relation between frontier gauges and fully
operational reliability functions should be sharpened. The
Bhattacharyya sandwich (Proposition~\ref{thm:bridge}) connects the
frontier to repetition-code reliability from both sides, but
the full Shannon reliability function allows arbitrary $n$-block
codes. Determining when the frontier gives the correct operational
gauge, and when it gives only a repetition-code gauge, is a separate
problem.

\emph{Third}, the exact unguarded multipath zero-rate diversity gauge remains
open. The guarded construction isolates independent delayed
scale-family observations and gives an endpoint DIV lower bound
proportional to the number of nonzero paths. What remains is to
decide whether unguarded coding can improve this gauge or coefficient,
or whether a matching converse forces the guarded lower bound to be
sharp after the appropriate block normalization.

\emph{Fourth}, the stationary Toeplitz zero-rate diversity gauge
remains open. The on--off log-det calculation gives only a two-point
candidate lower bound and does not by itself establish a
fractional-log zero-rate diversity gauge. The main obstruction is
noncommutativity: for different codewords, the covariance matrices
$I + \snr D_X R_T D_X^\dagger$ cannot generally be diagonalized in a
common basis. Closing this problem requires either a new comparison
principle for noncommuting Toeplitz covariance families or a matching
converse that identifies the correct binary endpoint scale.

\emph{Fifth}, the Toeplitz load-$r$ frontier is open. Even if the
zero-rate diversity gauge were identified, the positive-load problem
would still require packing many codewords under pairwise
noncommuting covariance constraints. This is the fractional-log
analogue of the fast-fading load-$r$ theorem, but without the
one-dimensional scale-family structure that made the fast-fading proof
exact.

\emph{Finally}, the framework should be extended beyond Gaussian
fading. The present paper relies heavily on Gaussian Bhattacharyya
formulas, product multiplicativity, covariance geometry, and
scale-family reductions. A broader theory would identify which parts
of gauge selection and atomic normalization survive for non-Gaussian
noise, non-Gaussian fading, peak-limited nonlinear channels, or
channels with memory beyond the Gaussian Toeplitz setting.
Appendix~\ref{app:non_gaussian} gives a simple non-Gaussian output-law example illustrating that the gauge-selection and endpoint-normalization steps are not intrinsically Gaussian, even though the main channel audit in this paper is restricted to Gaussian fading models.

\appendices
\numberwithin{lemma}{section}
\numberwithin{proposition}{section}

\section{Bhattacharyya Distance}\label{app:bhatt}

\subsection{Multiplicativity}

\begin{lemma}[Multiplicativity]\label{lem:mult}
If $P = P_1 \times P_2$ and $Q = Q_1 \times Q_2$, then
$B(P,Q) = B(P_1,Q_1)\,B(P_2,Q_2)$, i.e.,
$d_B(P,Q) = d_B(P_1,Q_1) + d_B(P_2,Q_2)$.
\end{lemma}
\begin{proof}
By Fubini,
$\int\sqrt{p_1 p_2}\sqrt{q_1 q_2}
= (\int\sqrt{p_1 q_1})(\int\sqrt{p_2 q_2})$.
\end{proof}

\subsection{Closed Forms for Gaussians}\label{app:bhatt_gaussian}

\begin{lemma}[Same covariance]\label{lem:bhatt_same_cov}
For $P = \mathcal{CN}(\mu_1,\Sigma)$,
$Q = \mathcal{CN}(\mu_2,\Sigma)$:
\begin{equation}\label{eq:bhatt_same_cov}
d_B(P,Q) = \frac{1}{4\ln 2}\,
(\mu_1-\mu_2)^\dagger\Sigma^{-1}(\mu_1-\mu_2).
\end{equation}
\end{lemma}
\begin{proof}
Let $\Delta\mu = \mu_1 - \mu_2$. For circular complex Gaussians,
$p(y) = (\pi^n\det\Sigma)^{-1}
\exp(-(y-\mu_1)^\dagger\Sigma^{-1}(y-\mu_1))$ (no factor of
$1/2$ in the exponent). Completing the square:
$\sqrt{p(y)\,q(y)} = c\,\exp\!\big(-(y-\bar\mu)^\dagger
\Sigma^{-1}(y-\bar\mu)
- \frac{1}{4}\Delta\mu^\dagger\Sigma^{-1}\Delta\mu\big)$
where $\bar\mu = (\mu_1+\mu_2)/2$.
Integrating over~$y$ yields
$B = \exp(-\frac{1}{4}\Delta\mu^\dagger\Sigma^{-1}\Delta\mu)$,
hence $d_B = \frac{1}{4\ln 2}
\Delta\mu^\dagger\Sigma^{-1}\Delta\mu$.
\end{proof}

\begin{lemma}[Same mean]\label{lem:bhatt_same_mean}
For $P = \mathcal{CN}(0,\Sigma_1)$,
$Q = \mathcal{CN}(0,\Sigma_2)$:
\begin{equation}\label{eq:bhatt_same_mean}
d_B(P,Q) = \log
\frac{\det\!\big(\tfrac{\Sigma_1+\Sigma_2}{2}\big)}
{\sqrt{\det\Sigma_1\cdot\det\Sigma_2}}.
\end{equation}
\end{lemma}
\begin{proof}
Write $p(y) = \pi^{-n}(\det\Sigma_1)^{-1}
\exp(-y^\dagger\Sigma_1^{-1}y)$ and similarly for~$q$.
Then $\sqrt{p\,q} \propto
\exp\!\big(-y^\dagger
\frac{\Sigma_1^{-1}+\Sigma_2^{-1}}{2}\,y\big)$.
This is an unnormalized Gaussian with covariance
$\bar\Sigma = ((\Sigma_1^{-1}+\Sigma_2^{-1})/2)^{-1}$.
The integral evaluates to
$B = \frac{\sqrt{\det\Sigma_1\cdot\det\Sigma_2}}
{\det((\Sigma_1+\Sigma_2)/2)}$.
Taking $d_B = -\log B$ gives the result.
\end{proof}

Lemma~\ref{lem:bhatt_same_cov} governs coherent channels
(known~$H$: location family);
Lemma~\ref{lem:bhatt_same_mean} governs noncoherent channels
(zero-mean outputs, different covariances). All $d_B$ expressions
use base-2 logarithms (consistent with
Section~\ref{sec:vocab}); the $1/(4\ln 2)$ factor in
Lemma~\ref{lem:bhatt_same_cov} arises from converting the
natural exponential to a base-2 log.

\subsection{Rank-Product Decomposition}\label{app:rank_product}

\begin{lemma}[Pairwise-difference decomposition and Rayleigh
averaging]\label{lem:rank_product}
Let $D = X_1 - X_2$ have nonzero singular values
$s_1,\dots,s_q$, where $q = \mathrm{rank}(D)$, and let
$u_1,\dots,u_q \in \C^M$ denote the corresponding left singular
vectors of~$D$.

\textbf{(a) Known channel (fixed~$H$): algebraic decomposition.}
By Lemma~\ref{lem:bhatt_same_cov}:
\begin{equation}\label{eq:rank_product_det}
d_B = \frac{\snr}{4\ln 2}\sum_{j=1}^{q}\sum_{\ell=1}^{N}
s_j^2\,|H_\ell^\dagger u_j|^2,
\end{equation}
a sum of $N\!\cdot\!\mathrm{rank}(D)$ terms, each proportional
to~$\snr$.

\textbf{(b) Rayleigh fading ($H_{ij}$ i.i.d.\
$\mathcal{CN}(0,1)$): rank-product diversity coefficient.}
By unitary invariance and averaging over~$H$:
\begin{equation}\label{eq:rank_product_fading}
\bar{d}_B = N\sum_{j=1}^{q}\log\!\Big(1+\frac{\snr\,s_j^2}{4}
\Big),
\end{equation}
a sum of $N\!\cdot\!\mathrm{rank}(D)$ terms, each $\sim\log\snr$.
\end{lemma}
\begin{proof}
(a)~Write $\|HD\|_F^2 = \sum_{\ell=1}^{N}\sum_{j=1}^{q}
s_j^2\,|H_\ell^\dagger u_j|^2$ and substitute
into $d_B = \frac{\snr}{4\ln 2}\|HD\|_F^2$
(Lemma~\ref{lem:bhatt_same_cov}).
(b)~By unitary invariance,
$\|HD\|_F^2 = \sum_j s_j^2\|g_j\|^2$ with
$g_j \sim \mathcal{CN}(0,I_N)$ independent.
$\E_H[B] = \prod_j(1+\snr s_j^2/4)^{-N}$.
Taking $\bar{d}_B = -\log\E_H[B]$
gives~\eqref{eq:rank_product_fading}.
\end{proof}

The rank product $N\!\cdot\!\mathrm{rank}(D)$ controls the number
of additive terms. For known~$H$, the decomposition is algebraic:
each term is proportional to $\snr\,s_j^2\,|H_\ell^\dagger
u_j|^2$, so individual term magnitudes depend on the alignment
of~$D$ with the singular directions of~$H$ (terms can vanish if
$D$ points into the nullspace or weak directions of~$H$). For
Rayleigh fading, the rank-product statement is much safer:
by unitary invariance, each of the $N\!\cdot\!\mathrm{rank}(D)$
output-law factors contributes $\sim\log\snr$ after averaging,
because the
deep-fade event $|H|^2 \approx 0$ dominates the average and
pulls each contribution from~$\snr$ down to~$\log\snr$.

\subsection{Determinant Submultiplicativity}

\begin{lemma}[Determinant submultiplicativity for positive semidefinite sums]
\label{lem:det_submult}
If $X,Y \succeq 0$, then
\begin{equation}\label{eq:det_submult}
\det(I+X+Y) \le \det(I+X)\det(I+Y).
\end{equation}
\end{lemma}
\begin{proof}
We have
$\det(I+X+Y) = \det(I+X)\det(I+(I+X)^{-1}Y)$.
The second determinant equals
$\det(I+Y^{1/2}(I+X)^{-1}Y^{1/2})$.
Since $0 \preceq (I+X)^{-1} \preceq I$, we have
$0 \preceq Y^{1/2}(I+X)^{-1}Y^{1/2} \preceq Y$.
Monotonicity of $\det(I+\cdot)$ on the positive semidefinite cone
gives the result.
\end{proof}

\subsection{Inverse Relationship}

\begin{lemma}[Inverse relationship]\label{lem:inverse}
$N_{\mathrm{pack}}(\delta;\snr) \ge K
\;\Longleftrightarrow\;
\Delta_B^*(K;\snr) \ge \delta$.
\end{lemma}
\begin{proof}
$(\Rightarrow)$:
If $N_{\mathrm{pack}}(\delta;\snr) \ge K$, there exists
$\setC \subseteq \setX_{\mathrm{one}}$ with $|\setC| \ge K$ and minimum
pairwise $d_B \ge \delta$. Any size-$K$ subset witnesses
$\Delta_B^*(K;\snr) \ge \delta$.
$(\Leftarrow)$:
If $\Delta_B^*(K;\snr) \ge \delta$, there exists $\setC$ with
$|\setC| = K$ and minimum pairwise $d_B \ge \delta$, so
$N_{\mathrm{pack}}(\delta;\snr) \ge K$.
\end{proof}

\section{Gauge Uniqueness}\label{app:gauge_unique}

The gauge vocabulary uses divergent functions in Hardy's
logarithmico-exponential (LE) class~\cite{Hardy1910}. Hardy's
comparison theorem totally orders this class: for any two divergent
LE functions $g_1,g_2$, exactly one of
\[
    g_1/g_2 \to 0,\quad
    g_1/g_2 \to \infty,\quad
    g_1/g_2 \to c \in (0,\infty)
\]
holds. Thus the gauges used in this paper---$\log\snr$,
$\log\log\snr$, $(\log\snr)^\beta$, $\snr$, and $\log\snr$ on the
diversity side---have unambiguous asymptotic classes. This uniqueness
convention is the only role of the Hardy LE assumption.

\section{Coding Bounds}\label{app:coding_bounds}

\subsection{Capacity-Side Bridge Proof}\label{app:sandwich_proof}

\begin{proof}[Proof of Theorem~\ref{lem:sandwich}]
For the lower bound, use the packing $x_1,\dots,x_{M_\snr}$ and put
the uniform distribution on these inputs. Write
$P_i = P_{x_i}^{(\snr)}$,
$\bar{P} = \frac{1}{M_\snr}\sum_{j=1}^{M_\snr} P_j$.
Then
\[
I(X;Y) = \frac{1}{M_\snr}\sum_{i=1}^{M_\snr}
D(P_i \| \bar{P}).
\]
Since KL divergence dominates R\'enyi divergence of order~$1/2$,
\[
D(P_i \| \bar{P}) \ge -2\log B(P_i, \bar{P}).
\]
Also,
\[
B(P_i, \bar{P})
= \int\!\sqrt{p_i \cdot \tfrac{1}{M_\snr}\sum_j p_j}
\le \frac{1}{\sqrt{M_\snr}}\sum_{j=1}^{M_\snr}
B(P_i, P_j).
\]
Hence
\[
D(P_i \| \bar{P}) \ge \log M_\snr
- 2\log\!\bigg(\sum_{j=1}^{M_\snr} B(P_i, P_j)\bigg).
\]
By the aggregate-overlap hypothesis,
\[
I(X;Y) \ge \log M_\snr - o(g(\snr))
\ge c_-\,g(\snr) - o(g(\snr)).
\]
Therefore
\[
\liminf_{\snr\to\infty}
\frac{C_{\mathrm{block}}(\snr)}{g(\snr)} \ge c_-.
\]

For the upper bound, take an arbitrary admissible input
distribution~$P_X$. Let $J$ be the index of the covering cell
containing~$X$. Since $J$ is a function of~$X$,
\[
I(X;Y) = I(J,X;Y) = I(J;Y) + I(X;Y|J).
\]
The first term satisfies
\[
I(J;Y) \le H(J) \le \log L_\snr
\le c_+\,g(\snr) + o(g(\snr)).
\]
For the second term, condition on $J=\ell$. For any auxiliary output
law~$Q_\ell$,
\[
I(X;Y|J=\ell)
\le \E\!\big[D(P_X^{(\snr)} \| Q_\ell) \mid J=\ell\big],
\]
because
\[
\E\!\big[D(P_X^{(\snr)} \| Q_\ell) \mid J=\ell\big]
= I(X;Y|J=\ell) + D(P_{Y|J=\ell} \| Q_\ell).
\]
Choosing $Q_\ell$ from the covering hypothesis gives
\[
I(X;Y|J) = o(g(\snr)).
\]
Thus
\[
I(X;Y) \le c_+\,g(\snr) + o(g(\snr)).
\]
Taking the supremum over~$P_X$,
\[
\limsup_{\snr\to\infty}
\frac{C_{\mathrm{block}}(\snr)}{g(\snr)} \le c_+.
\]
Combining the lower and upper bounds gives
\[
c_- \;\le\;
\liminf_{\snr\to\infty} \frac{C_{\mathrm{block}}(\snr)}{g(\snr)}
\;\le\;
\limsup_{\snr\to\infty} \frac{C_{\mathrm{block}}(\snr)}{g(\snr)}
\;\le\; c_+.
\]
If $c_- = c_+ = c$, then
\[
C_{\mathrm{block}}(\snr) \sim c\,g(\snr).
\]
If the sharper $O(1)$ packing, covering, overlap, and KL-radius
bounds hold, the same argument gives
\[
C_{\mathrm{block}}(\snr) = c\,g(\snr) + O(1).
\]
\end{proof}

\subsection{Bhattacharyya Sandwich Proof}

\begin{proof}[Proof of Proposition~\ref{thm:bridge}]
\emph{Upper bound.}
For message~$i$, the union--Bhattacharyya bound gives
$P_{e|i} \le \sum_{j\ne i} B(P_i,P_j) \le
(K-1)\,2^{-\Delta_{\min}(\setC;\snr)}$.
Averaging over~$i$ gives the upper bound. With $n$~uses,
multiplicativity (Lemma~\ref{lem:mult}) gives
$B^n = 2^{-n\,d_B}$.

\emph{Lower bound.}
Choose a closest pair $i,j$. Any $K$-ary decoder induces a binary
test between $P_i$ and $P_j$. By Theorem~\ref{lem:binary_endpoint},
the optimal binary error satisfies
$P_{e,2}^* \ge \tfrac{1}{4}\,2^{-2\,d_B(P_i,P_j)}$.
The two conditional errors for messages~$i$ and~$j$ contribute to the
average $K$-message error, giving the factor~$1/K$:
$P_{e,\mathrm{opt}} \ge \tfrac{1}{2K}\,2^{-2\,\Delta_{\min}(\setC;\snr)}$.
Repetition multiplies $d_B$ by~$n$.
\end{proof}

\section{Structural Mechanisms}\label{app:mechanisms}

\begin{proposition}[Scale-family packing]\label{prop:scale_family}
Consider the scale family
$P_v = \mathcal{CN}(0,v)$, $v \in [1,1+\snr]$,
with $N$~independent receive observations. Let $u = \ln v$. Then
\begin{equation}\label{eq:scale_bhatt_app}
d_B(P_{v_1}^{\otimes N}, P_{v_2}^{\otimes N})
= N\log_2\cosh\!\Big(\frac{|u_1-u_2|}{2}\Big).
\end{equation}
Therefore a fixed Bhattacharyya separation~$\delta$ is equivalent
to spacing
$|u_1-u_2| \ge c_N(\delta)$, where
$c_N(\delta) = 2\operatorname{arcosh}(2^{\delta/N})$.
Since $u \in [0,\ln(1+\snr)]$, the packing number satisfies
$N_{\mathrm{pack}}(\delta;\snr) \asymp \log\snr$, and hence
\begin{equation}\label{eq:kpack_loglog}
K_{\mathrm{pack}}(\delta;\snr) = \log N_{\mathrm{pack}}(\delta;\snr)
= \log\log\snr + O(1).
\end{equation}
\end{proposition}

\begin{proof}
The output family is $\{P_v = \mathcal{CN}(0,v) :
v \in [1,1+\snr]\}$, parameterized by $v = \snr|x|^2+1$. By
Lemma~\ref{lem:bhatt_same_mean}:
$d_B(P_{v_1},P_{v_2}) = \log\cosh((\ln v_1-\ln v_2)/2)$.
With $N$~independent observations, multiplicativity gives
$d_B = N\log\cosh(|u_1-u_2|/2)$.
Change coordinates to $u = \ln v \in [0,L]$ with
$L = \ln(1+\snr) \sim \ln\snr$.

\emph{Achievability.}
Place $K$ equally spaced points in $[0,L]$. Minimum $u$-separation
$\Delta_u = L/(K-1) \ge c_N(\delta)$ gives
$N_{\mathrm{pack}} \ge 1+L/c_N(\delta) \asymp \log\snr$, hence
$K_{\mathrm{pack}} \ge \log\log\snr + O(1)$.

\emph{Converse.}
Any codebook with $d_B \ge \delta$ has $u$-coordinates that are
$c_N(\delta)$-separated in $[0,L]$, giving at most
$1+L/c_N(\delta) \asymp \log\snr$ points, hence
$K_{\mathrm{pack}} \le \log\log\snr + O(1)$.
\end{proof}

\begin{remark}[Toeplitz on-off endpoint]\label{rem:toeplitz_onoff}
The on-off pair $x = \sqrt{P}\,\mathbf{1}_T$ and $x = \mathbf{0}$
gives a lower bound on $\Delta_B^*(2;\snr)$, not on
$K_{\mathrm{pack}}(\delta;\snr)$. A two-point construction can only
prove $N_{\mathrm{pack}} \ge 2$.
Per receive antenna, for the Toeplitz covariance~$R_T$,
\[
d_{B,1}^{(\text{on-off})}
= \log\det\!\Big(I_T + \tfrac{\snr P}{2}\,R_T\Big)
  - \tfrac{1}{2}\log\det\!\big(I_T + \snr P\,R_T\big).
\]
With $N$ independent receive antennas,
$d_{B,N}^{(\text{on-off})} = N\,d_{B,1}^{(\text{on-off})}$.
The associated Szeg\H{o} integral is
$I(a) = \frac{1}{2\pi}\int_{-\pi}^{\pi}\log(1+a\,f_H(\lambda))
\,d\lambda$.
For fixed nonzero $f_H \in L^1$, $I(a) = \Theta(\log a)$.
Thus this on-off calculation does not establish a fractional-log
diversity gauge. The stationary Toeplitz diversity gauge remains
open.
\end{remark}

\section{A Non-Gaussian Output-Law Example}\label{app:non_gaussian}

This appendix gives a simple non-Gaussian example showing that the frontier construction is not tied to Gaussian formulas. The example is not part of the Gaussian-fading channel audit in Tables III--IV; it is included only to illustrate that the same output-law method applies beyond Gaussian output laws.

Consider the real scalar additive channel
\[
Y=\sqrt{\snr}\,X+Z,\qquad |X|\le 1,
\]
where $Z$ has the unit-scale Laplace density
\[
f_Z(z)=\frac12 e^{-|z|}.
\]
Let $P_\mu$ denote the output law of $\mu+Z$, where $\mu=\sqrt{\snr}x$. Thus the admissible output-law family is the translation family
\[
\{P_\mu:\mu\in[-\sqrt{\snr},\sqrt{\snr}]\}.
\]

For two shifts $\mu,\nu$, write $a=|\mu-\nu|$. A direct calculation gives the Bhattacharyya coefficient
\[
B(P_\mu,P_\nu)
=
\int \sqrt{f_Z(y-\mu)f_Z(y-\nu)}\,dy
=
e^{-a/2}\left(1+\frac{a}{2}\right).
\]
Therefore
\[
d_B(P_\mu,P_\nu)
=
\frac{a}{2\ln 2}
-
\log\left(1+\frac{a}{2}\right),
\]
where $\log$ is base~$2$, as in the rest of the paper. Hence the Bhattacharyya distance is a monotone function of the shift separation~$a$.

For any fixed separation threshold $\delta>0$, let $a_\delta$ be the corresponding shift separation satisfying $d_B(a_\delta)=\delta$. Packing the interval $[-\sqrt{\snr},\sqrt{\snr}]$ at spacing $a_\delta$ gives
\[
N_{\mathrm{pack}}(\delta;\snr)\asymp \sqrt{\snr},
\]
and hence
\[
K_{\mathrm{pack}}(\delta;\snr)
=
\log N_{\mathrm{pack}}(\delta;\snr)
=
\frac12\log\snr+O(1).
\]

The covering side has the same coefficient. Cover the shift interval by intervals of fixed length. The local KL radius is bounded because, for $a\ge 0$,
\[
D(P_0\|P_a)
=
\frac{a+e^{-a}-1}{\ln 2}.
\]
Thus a fixed-length shift cell has bounded local KL radius. Also, for a fixed-spacing grid, the aggregate Bhattacharyya overlap is uniformly bounded since
\[
\sum_{k\ge 1}
B(P_0,P_{ka})
=
\sum_{k\ge 1}
e^{-ka/2}\left(1+\frac{ka}{2}\right)
<\infty.
\]
The capacity-side bridge, Theorem~1, therefore gives
\[
C(\snr)=\frac12\log\snr+O(1).
\]
Thus the capacity gauge is $\log\snr$, with raw coefficient $1/2$. Under the usual real one-dimensional communication-atom normalization, this corresponds to one real communication atom.

At the binary endpoint, the maximum shift separation is attained by the antipodal pair $x_1=-1$, $x_2=1$, giving $a=2\sqrt{\snr}$. Therefore
\[
\Delta_B^*(2;\snr)
=
d_B(2\sqrt{\snr})
=
\frac{\sqrt{\snr}}{\ln 2}
-
\log(1+\sqrt{\snr})
\sim
\frac{1}{\ln 2}\sqrt{\snr}.
\]
Hence the zero-rate diversity gauge is $\sqrt{\snr}$, not $\snr$.

For this antipodal pair, the optimal equal-prior binary test is threshold detection at zero. Its error probability is
\[
P_{e,2}^*
=
\frac12 e^{-\sqrt{\snr}},
\]
so
\[
-\log P_{e,2}^*
=
\frac{\sqrt{\snr}}{\ln 2}+1.
\]
Thus the binary Bhattacharyya endpoint identifies the actual two-message zero-rate reliability gauge, and in this example it also recovers the leading raw operational coefficient.

This example shows that the gauge-selection step is output-law dependent. The rate side is controlled by fixed-resolution packing of a one-dimensional translation interval and therefore lives on the $\log\snr$ gauge. The binary reliability side is controlled by the tail of the noise distribution; for Laplace noise it lives on the $\sqrt{\snr}$ gauge. The same frontier construction therefore extends beyond Gaussian formulas: Gaussianity is used in the main text to obtain explicit channel-by-channel evaluations, but the underlying operations are output-law packing, gauge selection, and endpoint normalization.

\bibliographystyle{IEEEtran}
\bibliography{refs}

\end{document}